\documentclass[12pt]{article}
\usepackage{amsmath}
\usepackage{graphicx}
\usepackage{enumerate}
\usepackage{natbib}
\usepackage{url} 
\usepackage{multirow}
\usepackage{colortbl}
\usepackage{xcolor}
\usepackage{bbm}
\usepackage{rotating}

\newcommand{\blind}{1}

\begin{document}

\def\spacingset#1{\renewcommand{\baselinestretch}%
{#1}\small\normalsize} \spacingset{1}
\def\bs{\boldsymbol}
\def\bl{\textbf}



%

\newtheorem{theorem}{Theorem}[section]
\newenvironment{proof}[1][Proof]{\begin{trivlist}
\item[\hskip \labelsep {\bfseries #1}]}{\end{trivlist}}

\if1\blind
{
  \title{\bf Likelihood Inference for Possibly Non-Stationary Processes via Adaptive Overdifferencing}
  \author{Maryclare Griffin\thanks{
    Financial support is gratefully acknowledged from a Xerox PARC Faculty Research Award, National Science Foundation Awards 1455172, 1934985, 1940124, 2310974, and 1940276, USAID, and Cornell University Atkinson Center for a Sustainable Future. Address for correspondence: Maryclare Griffin, Lederle Graduate Research Tower, North Pleasant Street, Amherst, MA 01003, USA. Email: maryclaregri@umass.edu}\hspace{.2cm}\\
    Department of Mathematics and Statistics \\ University of Massachusetts Amherst, Amherst, MA, USA\\
    Gennady Samorodnitsky\\
    School of Operations Research and Information Engineering \\ Cornell University, Ithaca, NY, USA and \\
   David S. Matteson \\
    Department of Statistics and Data Science \\ Cornell University, Ithaca, NY, USA}
  \maketitle 
} \fi

\if0\blind
{
  \bigskip
  \bigskip
  \bigskip
  \begin{center}
    {\LARGE\bf Likelihood Inference for Possibly Non-Stationary Processes via Adaptive Overdifferencing}
\end{center}
  \medskip
} \fi

\bigskip
\begin{abstract}
We make an observation that facilitates exact likelihood-based inference for the parameters of the popular ARFIMA model without requiring stationarity by allowing the upper bound $\bar{d}$ for the memory parameter $d$ to exceed $0.5$: estimating the parameters of a single non-stationary ARFIMA model is equivalent to estimating the parameters of a sequence of stationary ARFIMA models. This allows for the use of existing methods for evaluating the likelihood for an invertible and stationary ARFIMA model.
This enables improved inference because many standard methods perform poorly when estimates are close to the boundary of the parameter space. It also allows us to leverage the wealth of likelihood approximations that have been introduced for estimating the parameters of a stationary process. 
We explore how estimation of the memory parameter $d$ depends on the upper bound $\bar{d}$ and introduce adaptive procedures for choosing $\bar{d}$. 
We show via simulation how our adaptive procedures estimate the memory parameter well, relative to existing alternatives, when the true value is as large as 2.5.
\end{abstract}

\noindent%
{\it MOS subject classification:} 62M10
\vfill
\newpage
\spacingset{1.5} 
\section{Introduction}

Many methods for analyzing  an equally spaced time series $\bs y = \left(y_1, \dots, y_n\right)$  have been developed. Stationary autoregressive moving average (ARMA) models and their non-stationary generalizations predominate. An  ARMA$\left(p, q\right)$ model assumes that deviations  of  observations  $y_t$  from  their  means  $\mu_t$  are  a function of past  deviations  and stochastic errors,
\begin{align}\label{eq:arma}
\phi\left(B\right) \left(y_t - \mu_{t}\right) = \theta\left(B\right) z_t\text{, \quad} z_t \stackrel{i.i.d.}{\sim}\mathcal{N}\left(0, \sigma^2\right),
\end{align}
where $\phi\left(B\right) = 1 - \sum_{\ell = 1}^p \phi_{\ell} B^{\ell}$, $\theta\left(B\right) = 1 + \sum_{\ell = 1}^q \theta_{\ell} B^{\ell}$, and $B$ is the shift $B^{\ell} y_t = y_{t-\ell}$.  The mean $\mu_t$ is a specified function of a small number of  unknown  parameters, and may  be  assumed to be  an unknown  constant, a low degree polynomial function in time $t$  with unknown coefficients, or a linear function of a small number of predictors  with unknown coefficients.  Equation~\eqref{eq:arma}  describes a stationary, causal, and invertible model   when all roots of the autoregressive polynomial $\phi\left(z\right)$ and all roots of the moving average polynomial $\theta\left(z\right)$ lie outside of the unit circle. Stationarity ensures that the mean and variance of the  deviations $y_t - \mu_t$  are  constant over time and  that  correlations between deviations  depend only on how far apart they are in time. \color{black}T\color{black}he autocorrelation function uniquely determines the ARMA$\left(p, q\right)$ parameters. 
An autoregressive integrated moving average (ARIMA) model generalizes the ARMA model  to allow for certain types of non-stationarity, specifically the  presence of certain deterministic time trends \citep{Box1970}.  An ARIMA$\left(p, d, q\right)$ model assumes that there is a nonnegative integer $d$ such that the $d$-th differences $\left(1 - B\right)^d \left(y_t - \mu_{t}\right)$ satisfy the ARMA$\left(p, q\right)$ equation.  Thus, the ARIMA$\left(p, d, q\right)$ model is equivalent to assuming  an ARMA$\left(p, q\right)$ model  for deviations of a simple function of observations from their means $\left(1 - B\right)^d\mu_t$, which can allow for the presence of certain deterministic trends without estimating them because $\left(1 - B\right)^d\mu_t = 0$ when $\mu_t = \sum_{j = 0}^{d-1} t^j \lambda_j$. 

Stationary ARMA$\left(p, q\right)$ models are not well suited  for modeling  correlations that decay very slowly over time  because slowly decaying correlations require ARMA$\left(p, q\right)$ models with an increasing number of parameters as the length of the time series grows \citep{Granger1980a}. 
For this purpose, long memory or autoregressive fractionally differenced moving average (ARFIMA or FARIMA) models have been developed \citep{Hosking1981, Granger1980a}. 
 An ARFIMA$\left(p, d, q\right)$  model assumes that a fractional difference of the deviations  $\left(1 - B\right)^d \left(y_t - \mu_{t}\right)$ is distributed according to a stationary ARMA$\left(p, q\right)$ model where
\begin{align}\label{eq:expandfd}
\left(1 - B\right)^d =\sum_{\ell = 0}^\infty {d \choose \ell}\left(-1\right)^{\ell} B^{\ell}.
\end{align}
This describes a stationary and invertible ARFIMA$\left(p, d, q\right)$ model when $-1 < d < 0.5$ \citep{Odaki1993}. 
An  ARFIMA$\left(p, d, q\right)$ process has slowly decaying correlations  when $d > 0$,  because  each  deviation  $y_t - \mu_t$  depends on infinitely many past  deviations and the corresponding weights decay slowly. The ARIMA$\left(p, d, q\right)$ model is obtained when $d$ is an integer.
Such models are fit to the data by first differencing the data an appropriate number of times, and fitting a stationary model to the differenced  deviations.

One common approach for maximum likelihood estimation of possibly non-stationary  ARFIMA$\left(p, d, q\right)$ models is to find the smallest integer $k$ for which $\left(1 - B\right)^k \left(y_t - \mu_{t}\right)$ appears to be stationary and  assume a stationary ARFIMA$\left(p, d, q\right)$ model  for  the differenced  deviations  $\left(1 - B\right)^k \left(y_t - \mu_{t}\right)$ for $-0.5 < d  < 0.5$.  This approach is described in \cite{Hualde2011} and intuitively called the ``difference-and-add-back'' approach by \cite{Johansen2016}; it has been recommended by \cite{Box-Steffensmeier1998} and used in practice \citep{Byers1997,Byers2000,Dolado2003}.  
Although useful, this procedure can lead to practical challenges in the presence of nearly non-stationary differenced ARFIMA$\left(p, d, q\right)$ processes, which are processes that are well represented by values of $d$ close to the boundary of stationarity, e.g.\ $d \approx 0.5$, $d \approx 1.5$, or $d \approx 2.5$. 
\color{black}This is especially apparent when allowing  $d$ to vary across subsets of a single time series, e.g. as described  in \cite{Graves2015},  or when pooling information across multiple replicate time series to estimate a common $d$. The ``difference-and-add-back'' approach may lead to different amounts of differencing for different subsets of the same time series or for different time series. 
Furthermore, it can be difficult to decide whether or not a stationary model is reasonable for an observed time series with slowly decaying empirical correlations.

Figure~\ref{fig:toyex} illustrates this with two time series of length $n = 500$. Both time series are simulated according to an ARFIMA$\left(0, d, 0\right)$ model $\left(1 - B\right)^d y_t = z_t$ using the same stochastic errors $z_t$ simulated from a standard normal distribution. 
The first is simulated according to a stationary process with $d = 0.45$ and the second is simulated according to a non-stationary process with $d = 0.55$.
 The latter time series is simulated by taking cumulative sums of time series simulated according to a stationary process. 
Although the first time series is simulated from a stationary model and the second is not, both observed time series and their corresponding sample autocorrelation functions look similar. 

\begin{figure}[ht!]
\centering
\includegraphics{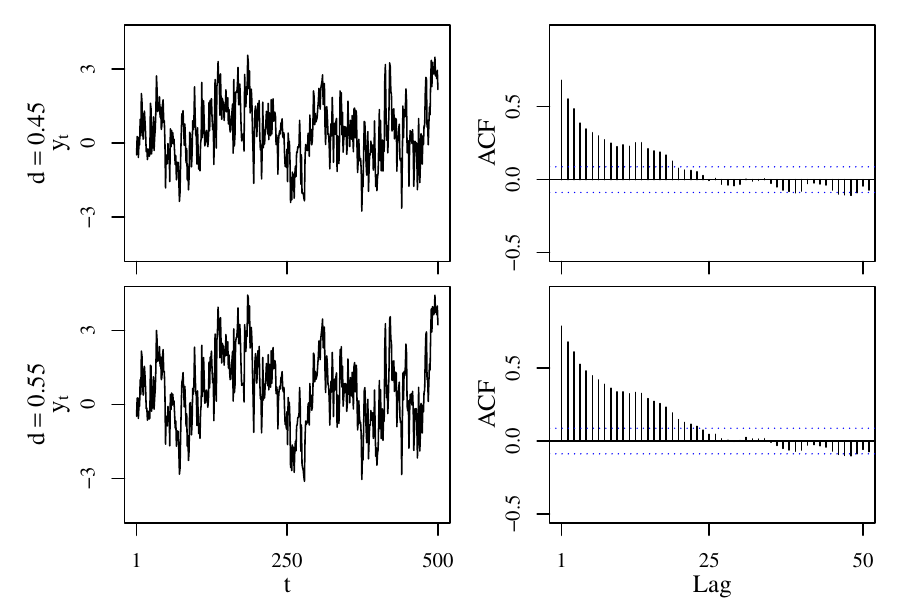}
\caption{Simulated length $n = 500$ time series and their sample autocorrelation functions (ACFs). Both time series satisfy the ARFIMA$\left(0, d, 0\right)$ model with $\mu = 0$ and the same stochastic errors with memory parameter $d = 0.45$ or $d = 0.55$. For reference, approximate 95\% intervals for sample autocorrelations of a white noise process are provided with ACFs.}
\label{fig:toyex}
\end{figure}

\begin{figure}[ht!]
\centering
\includegraphics{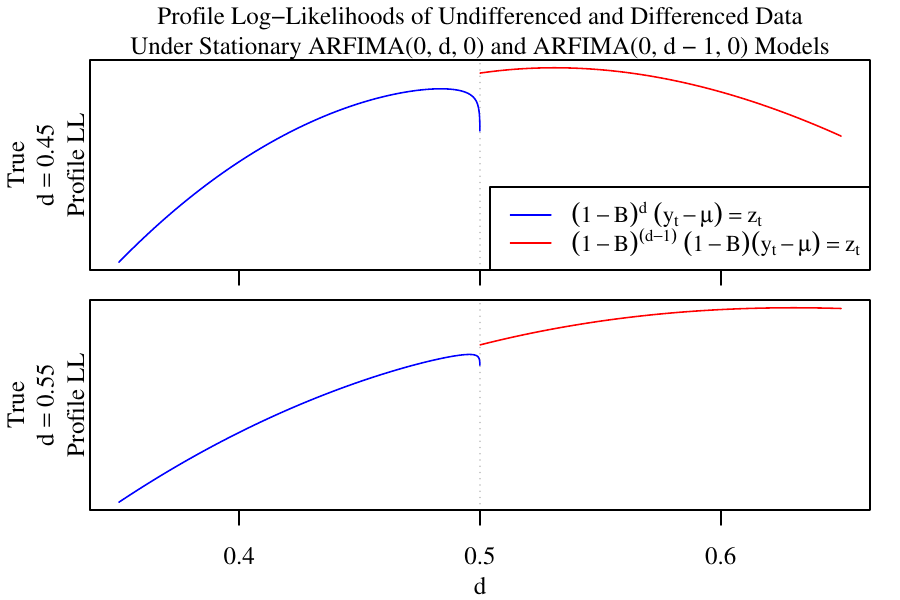}
\caption{Profile log-likelihoods of undifferenced and first-differenced data under an ARFIMA$\left(0, d, 0\right)$ with $\mu_t = \mu$  and  $\sigma^2$ profiled for the time series shown in Figure~\ref{fig:toyex}. 
}
\label{fig:lls}
\end{figure}

Despite the fact that the stationary and non-stationary time series shown in Figure~\ref{fig:toyex} look similar, the likelihoods of the two time series under an ARIMA$\left(0, d - k, 0\right)$ model  for the deviations $y_t - \mu$  obtained using  the  ``difference-and-add-back'' approach  of fitting a stationary ARIMA$\left(0, d - k, 0\right)$ model to the $k$-th differenced data $\left(1 - B\right)^k \left(y_t - \mu\right)$ for $-0.5 < d - k  < 0.5$ are not continuous at  $d=k + 0.5$;  see Figure~\ref{fig:lls}.
This is because the data changes from the $n$ observed time series values to the $n - 1$ observed differences when we evaluate the likelihood for $d > 0.5$.
For this reason, likelihood values obtained in this way are of limited utility. They are only comparable across subsets of $d$ values that correspond to the same amount of differencing. 
Furthermore, the presence of boundaries and discontinuities can produce misleading standard errors and confidence intervals for estimates of the memory parameter $d$, whether they are based on asymptotic or bootstrap methods. A possible solution is suggested by recognizing that $\left(1 - B\right)^d\left(y_t - \mu\right) = \left(1 - B\right)^{d-k} \left(1 - B\right)^k \left(y_t - \mu\right)$. Thus, if $y_t - \mu$ is an ARFIMA$\left(0, d, 0\right)$ process, then the $k$-th difference $\left(1 - B\right)^k \left(y_t - \mu\right)$ is an ARFIMA$\left(0, d - k, 0\right)$ process. However, an ARFIMA$\left(0, d - k, 0\right)$ process is only invertible  for $d > k - 1$. Although consistency and asymptotic normality of the maximum likelihood estimator of the memory parameter $d$ has been shown for all $d < k + 0.5$ \citep{Lieberman2012}, careful examination of the references in and of \cite{Lieberman2012} yields no examples of maximum likelihood estimators that allow for $d - k$ outside of the invertible range, $d \leq k  -1$. 

Many existing alternative solutions approximate the likelihood. 
The conditional sum-of-squares (CSS) approximate likelihood, as described in \cite{Beran1995}, \cite{Hualde2011}, and  \cite{Hualde2020}, uses the truncated difference $\left(1 - B\right)^d_+ \left(y_t - \mu\right) =\sum_{\ell = 0}^{t-1} {d \choose \ell}\left(-1\right)^{\ell-1} B^{\ell}\left(y_t - \mu\right)$ to approximate $\left(1 - B\right)^d \left(y_t - \mu\right)$ and obtain an approximate likelihood that is continuous in $d$.
\cite{Hualde2011} and \cite{Hualde2020} proved that the CSS approximate likelihood provides consistent, asymptotically normal parameter estimates under a model where the truncated fractional difference of the data $\left(1 - B\right)^d_+ \left(y_t - \mu_t\right) =\sum_{\ell = 0}^{t-1} {d \choose \ell}\left(-1\right)^{\ell-1} B^{\ell}\left(y_t - \mu_t\right)$ is distributed according to a stationary ARMA$\left(p, q\right)$ model with $\mu_t = \mu_0 \left(t\mathbbm{1}_{\left\{t \geq 1 \right\}}\right)^{\gamma_0}$ for unknown $\mu_0$ and $\gamma_0$.
However, CSS approximate likelihood based estimates can be more biased than exact likelihood based estimates in finite samples if the data is generated according to an ARFIMA$\left(p, d, q\right)$ process \citep{Johansen2016}.
Other  alternatives involve the specification of additional tuning parameters. 
 \cite{Velasco2000} introduced spectral methods for $-0.5 < d$ and \cite{Hurvich1998} introduced spectral methods that have the added benefit of being invariant to the presence of linear trends for $-0.5 < d < 1.5$. The estimators introduced in \cite{Velasco2000} depend on the tapering applied to the sample periodogram and both the estimators introduced in both \cite{Velasco2000} and \cite{Hurvich1998} require specification of the number of periodogram ordinates used for estimation. 
\cite{Mayoral2007} introduced a moment-based method for $d > -0.75$ based on the first $k$ sample autocorrelations, however it  requires specification of $k$.
 None of these references include comparisons to exact likelihood-based estimators that allow for estimates of the memory parameter outside of the stationary and invertible range.

This paper shows  that given an upper bound $\bar{d}$ for the memory parameter $d$, it is possible to implement   exact  likelihood  estimation  for differenced data for \emph{all} $d < \bar{d}$. 
 Our approach is motivated by the earlier observation that the $k$-th difference of an ARFIMA$\left(0, d, 0\right)$ process is a ARFIMA$\left(0, d - k, 0\right)$ process  and the literature showing consistency and  asymptotic normality of exact likelihood estimators of  ARFIMA$\left(0, d-k, 0\right)$ models for all $d -k< 0.5$ \citep{Lieberman2012}. 
Given an upper bound $\bar{d}$, we can difference the data before estimation and reduce the problem of estimating the parameters of a single non-stationary ARFIMA$\left(p, d, q\right)$ model to the problem of estimating the parameters of a sequence of stationary ARFIMA$\left(p, d, q\right)$ models with constrained moving average parameters. 
The upper bound $\bar{d}$  can be selected adaptively without a priori knowledge of the process' stationarity.

This paper proceeds as follows. First, we explain how the problem of estimating the parameters of a possibly non-stationary ARFIMA$\left(p, d, q\right)$ model with memory parameter $d$ bounded above by a fixed value $\bar{d}$ can be transformed to a simpler problem of estimating the parameters of a sequence of stationary ARFIMA$\left(p, d, q\right)$ models with constrained moving average parameters that correspond to  a  non-invertible moving average process, with likelihoods that can then be related to the likelihoods of invertible and stationary ARFIMA$\left(p, d, q\right)$ models.
We then introduce adaptive procedures for choosing $\bar{d}$.
We demonstrate the need for and performance of the adaptive procedures based on the exact likelihood and the approximate likelihoods  for ARFIMA$\left(0, d, 0\right)$ processes  in simulations. We show that the adaptive procedures can produce estimates of $d$ with low bias and, when based on the exact likelihood, confidence intervals with nominal coverage. We also compare the bias of adaptive exact likelihood estimators to the two alternatives  described in \cite{Beran1995} and introduced in \cite{Mayoral2007}. We observe comparable performance to alternatives  when $n$ is relatively small and better performance than the alternatives  as $n$ increases when long memory is present.
We  use the proposed methods to fit possibly non-stationary ARFIMA$\left(1, d, 0\right)$ and ARFIMA$\left(0, d, 1\right)$ models to time series featured in the literature, and discuss the challenges of estimating the parameters of possibly non-stationary ARFIMA$\left(p, d, q\right)$ models with $p > 0$ or $q > 0$ in practice. We apply the proposed methods to an existing problem which uses tests of non-stationarity of ARFIMA$\left(0, d, 0\right)$ models for deviations from a linear trend to assess mean reversion of OECD countries' per capita CO$_2$ emissions. Last, we apply the proposed methods to pooled estimation of ARFIMA$\left(0, d, 0\right)$ models from electric cell-substrate impedance sensing (ECIS) measurements.

\section{Methodology}

\subsection{Relating Non-Stationary to Stationary Problems Given $\bar{d}$}

Let $y_t-\mu_t$ be a possibly non-stationary ARFIMA$\left(p, d, q\right)$ process with autoregressive parameters $\bs \phi = \left(\phi_1, \dots, \phi_p\right)$ and moving average parameters $\bs \theta = \left(\theta_1, \dots, \theta_q\right)$ satisfying
\begin{align}\label{eq:arfima}
\phi\left(B\right) \left(1 - B\right)^d \left(y_t - \mu_{t} \right)= \theta\left(B \right) z_t\text{, \quad} z_t \stackrel{i.i.d.}{\sim}\mathcal{N}\left(0, \sigma^2\right),
\end{align} 
where all roots of $\phi\left(z\right)$ and $\theta\left(z\right)$ lie outside of the unit circle and $d < \bar{d}$ for some $\bar{d} \geq 0.5$.  
Deviations of the differenced process $x^{\left(m\right)}_t = \left(1 - B\right)^{m} y_t$  from their means $\mu^{\left(m\right)}_t = \left(1 - B\right)^{m} \mu_{t}$ are stationary  for any $m$ satisfying $\bar{d} - m < 0.5$ and can be computed exactly for $t > m$ if $m$ is an integer. Let $m_{\bar{d}}$ be the smallest integer that satisfies $\bar{d} - m_{\bar{d}} \leq 0.5$. 

Given an integer $m_{\bar{d}}$ satisfying  $\bar{d} - m_{\bar{d}}  \leq 0.5$,  we  can evaluate  the likelihood of the differenced deviations  $l_{\bar{d}}\left(\bs x^{\left(m_{\bar{d}}\right)}| d, \mu_t, \sigma, \bs \theta, \bs \phi\right)$ %
  for $d < \bar{d}$ \citep{Hosking1981}. When the differenced mean is constant, $\mu^{\left(m\right)}_t  = \mu$, the maximum likelihood estimator of $d$ based on the differenced deviations will be consistent and asymptotically normal for $d < \bar{d}$ \citep{Lieberman2012}.
This can produce  estimates of the  memory  parameter over the interval $d <  \bar{d}$ that are invariant to polynomial trends of degree $m$ or lower, depending on the assumed mean $\mu_t$. Setting $\bar{d} = 1.5$ and $\mu_t = \lambda_0 + \lambda_1 t$  yields  invariance to linear trends, and setting $\bar{d} = 2.5$ and $\mu_t =  \sum_{j = 0}^2 t^j \lambda_j$  yields  invariance to quadratic trends. This is beneficial when polynomial trends are a nuisance and limiting \color{black} when \color{black} polynomial trends are of interest. 

Despite theoretical justifications of maximum likelihood estimation for $d < \bar{d}$ \citep{Lieberman2012}, maximum likelihood estimation of a stationary ARFIMA$\left(p, d, q\right)$ process tends to require $\bar{d} - 1 \leq d < \bar{d}$, see e.g. \citep[pages 539-542]{Pipiras2017} and \cite{Durham2019}. Then the likelihood is only evaluated for parameters corresponding to an invertible process and efficient methods for evaluating the autocovariances which require $\bar{d} - 1 \leq d < \bar{d}$, e.g.  the methods of \cite{Sowell1992a}, can be used. This alleviates the computational burdens of evaluating autocovariances of a stationary ARFIMA$\left(p, d, q\right)$ process \citep{Doornik2003}.  
This motivates rewriting the model as 
\begin{align}\label{eq:arfimacon}
\phi\left(B\right) \left(1 - B\right)^{d - m_{\bar{d}} + j}\left(x^{\left(m_{\bar{d}}\right)}_t - \mu^{\left(m_{\bar{d}}\right)}_t\right)=
 \left(1 - B\right)^j \theta\left(B \right) z_t\text{, \quad} z_t \stackrel{i.i.d.}{\sim}\mathcal{N}\left(0, \sigma^2\right),
\end{align}
where $j$ is a nonnegative integer satisfying $-0.5 \leq d - m_{\bar{d}} + j < 0.5$. 
This is a stationary ARFIMA$\left(p, d - m_{\bar{d}} + j, j + q\right)$ model  for the differenced  deviations  with $j + q$ constrained moving average parameters $\tilde{\theta}_1^{\left(j\right)},\dots, \tilde{\theta}^{\left(j\right)}_{j + q}$  obtained by expanding out $\left(1 - B\right)^j\theta\left(B\right)$. It is non-invertible for $j > 0$; the moving average polynomial $\tilde{\theta}\left(B\right)$ has roots on the unit circle.

The advantage of rewriting the model is that the covariance $\boldsymbol \Gamma_{n - m_{\bar{d}}}\left(d - m_{\bar{d}} + j, \tilde{\bs \theta},\bs \phi, \sigma\right)$ of $n - m_{\bar{d}}$ differenced deviations under the ARFIMA$\left(p, d - m_{\bar{d}} + j, j + q\right)$ model \eqref{eq:arfimacon} is
\begin{align}\label{eq:arfimaconcov}
\boldsymbol \Gamma_{n - m_{\bar{d}}}\left(d - m_{\bar{d}} + j, \tilde{\bs \theta},\bs \phi, \sigma\right) = \bs A_{\bar{d}\left(j\right)} \boldsymbol \Omega_{n - m_{\bar{d}} + j}\left(d - m_{\bar{d}} + j, \bs \theta,\bs \phi, \sigma\right) \bs A_{\bar{d}\left(j\right)}',
\end{align}
where $\bs A_{\bar{d}\left(j\right)}$ is the $\left(n -   m_{\bar{d}}\right) \times \left(n - m_{\bar{d}} + j\right)$ $j$-th differencing matrix that returns $n - m_{\bar{d}} $ $j$-th differences and $ \boldsymbol \Omega_{n - m_{\bar{d}} + j}\left(d - m_{\bar{d}} + j, \bs \theta,\bs \phi, \sigma\right)$ is the $\left(n  - m_{\bar{d}} + j\right)\times \left(n  - m_{\bar{d}} + j\right)$ covariance matrix for a stationary and invertible ARFIMA$\left(p, d - m_{\bar{d}} + j, q\right)$ process with moving average and autoregressive parameters $\boldsymbol \theta$ and $\boldsymbol \phi$ and standard deviation $\sigma$. As a result, any method for obtaining autocovariances of a stationary and invertible ARFIMA$\left(p, d, q\right)$ process can be used to obtain the autocovariances needed to compute the likelihood of $\boldsymbol x^{\left(m_{\bar{d}}\right)} - \boldsymbol \mu^{\left(m_{\bar{d}}\right)}$. Furthermore, autocovariances can be reused across multiple values of $j$.

Given an upper bound $\bar{d}$  the ARFIMA$\left(p, d, q\right)$ likelihood can be obtained as a function of the mean $\mu_t$,  memory  parameter $d$, moving average and autoregressive autoregressive parameters $\boldsymbol \theta$ and $\boldsymbol \phi$, and standard deviation $\sigma$ by finding the integer $m_{\bar{d}}$ which satisfies $-0.5 < m_{\bar{d}} - \bar{d} \leq 0.5$, computing the differences $x^{\left(m_{\bar{d}}\right)}_t = \left(1 - B\right)^{m_{\bar{d}}} y_t$  and $\mu^{\left(m_{\bar{d}}\right)}_t = \left(1 - B\right)^{m_{\bar{d}}}\mu_t$, finding the integer $j > 0$ satisfying $-0.5 \leq d + m_{\bar{d}} + j  \leq 0.5$, and evaluating the likelihood of the deviations $x^{\left(m_{\bar{d}}\right)}_t - \mu^{\left(m_{\bar{d}}\right)}_t$ under the stationary constrained ARFIMA$\left(p, d - m_{\bar{d}} + j, j + q\right)$ model given by \eqref{eq:arfimacon}.	
This  yields a likelihood that is continuous for $d < \bar{d}$.  This is shown in Section~\ref{appsec:sec1} of the Appendix.

This procedure is amenable to the use of arbitrary approximations for ARFIMA$\left(p, d, q\right)$ likelihoods, specifically approximations which require stationarity. However, we caution that the resulting approximate likelihood for $d < \bar{d}$ may not be continuous.
\cite{Pipiras2017} provide a review of likelihood approximations for ARFIMA$\left(p, d, q\right)$ processes. 
We consider the Whittle approximation, which is obtained by substituting the Whittle log-likelihood, as described in \cite{Beran1995}, for the exact log-likelihood of an ARFIMA$\left(p, d - m_{\bar{d}} + j, j + q\right)$ process. This is equivalent to substituting the Whittle log-
likelihood for the exact log-likelihood of the ARFIMA$\left(p, d - m_{\bar{d}}, q\right)$ process, and produces 
a continuous approximate likelihood in $d < \bar{d}$.
We consider an additional conditional sum-of-squares approximation in Section~\ref{appsec:sec2} of the Appendix. 

\subsection{Data-Adaptive Choice of Upper Bound $\bar{d}$}\label{subsec:adapt}

It is desirable to set $\bar{d}$ as small as possible  because taking the $m_{\bar{d}}$-th difference $x^{\left(m_{\bar{d}}\right)}_t$ reduces the amount of available data from $n$ observations to $n - m_{\bar{d}}$ observations , prevents estimation of polynomial trends of degree $m$ or lower, and can lead to larger standard errors and wider confidence intervals.   For instance, if we choose the value $\bar{d} = 2.5$, then it is necessary to use $m_{\bar{d}} = 2$. This effectively yields $n - 2$ observations. At the same time, it is desirable to set $\bar{d}$ large enough that the estimator $\hat{d}_{\bar{d}}$ is not too close to the boundary $\bar{d}$, because many likelihood- and bootstrap-based methods for estimating standard errors for $\hat{d}_{\bar{d}}$, estimating confidence intervals for $\hat{d}_{\bar{d}}$, and performing tests using $\hat{d}_{\bar{d}}$ are likely to perform poorly when the estimator $\hat{d}_{\bar{d}}$ is close to the boundary $\bar{d}$.

Motivated by the asymptotic normality of the maximum likelihood estimate of the memory parameter under stationary ARFIMA$\left(p, d, q\right)$ models \citep{Lieberman2012},  we suggest the following adaptive procedure for selecting $\bar{d}$. 
We adaptively find the smallest $\bar{d}$ such that both the maximizing value $\hat{d}_{\bar{d}}$ and the approximate $\left(1 - \epsilon\right)\times 100$-percentile of the  asymptotic  distribution of $\hat{d}_{\bar{d}}$ is less than $\bar{d}$ for some small $\epsilon > 0$. We achieve the former by checking whether the profile likelihood of the differenced data $l_{\bar{d}}\left(\bs x^{\left(m_{\bar{d}}\right)} | d \right)$ is decreasing as $d$ approaches the upper boundary $\bar{d}$.
Crucially, this procedure avoids maximization of the log-likelihood for values of $\bar{d}$ that have a local maximum at the boundary, and is a stepwise procedure that stops once a suitable value $\bar{d}$ is reached.

Starting with $\bar{d} = 0.5$ and given $\delta > 0$ and $\epsilon > 0$, our procedure proceeds as follows:
\begin{enumerate}
	\item Approximate the derivative of the profile log-likelihood $l_{\bar{d}}\left(\bs x^{\left(m_{\bar{d}}\right)}|d \right)$ 
 at $d = \bar{d} - 2\delta$.
	\begin{align*}
	l_{\bar{d}}'\left(\bs x^{\left(m_{\bar{d}}\right)}|\bar{d} - 2\delta \right) \approx \frac{l_{\bar{d}}\left(\bs x^{\left(m_{\bar{d}}\right)}|\bar{d} - \delta \right) - l_{\bar{d}}\left(\bs x^{\left(m_{\bar{d}}\right)}|\bar{d} - 2\delta  \right)}{\delta}.
	\end{align*}
	\begin{itemize}
		\item If $l_{\bar{d}}'\left(\bs x^{\left(m_{\bar{d}}\right)}|\bar{d} - 2\delta\right) > 0$, set $\bar{d} = \bar{d} + 1$ and return to 1. Otherwise,  proceed to 2. 
		\end{itemize}
	\item Compute $\hat{d}_{\bar{d}}$, the value of the memory parameter $d$ that maximizes the log-likelihood $l_{\bar{d}}\left(\bs x^{\left(m_{\bar{d}}\right)}|d \right)$, and compute an approximate $\left(1-\epsilon\right)\times100 \%$ percentile according to $\hat{d}_{\bar{d}} + z_{1 - \epsilon} /\sqrt{l_{\bar{d}}''\left(\bs x^{\left(m_{\bar{d}}\right)}|\hat{d}_{\bar{d}} \right)\left(n - m_{\bar{d}} - p - q\right)}$, where $z_{1 - \epsilon}$ is the $\left(1 - \epsilon\right)\times 100 \%$ percentile of a standard normal distribution.
	\begin{itemize}
		\item If the $\left(1 - \epsilon\right)\times100 \%$ percentile exceeds $\bar{d}$,  set $\bar{d} = \bar{d} + 1$ and return to 1.
		\item If the $\left(1 - \epsilon\right)\times100 \%$ percentile is less than $\bar{d}$,  stop.
	\end{itemize}
\end{enumerate}

This creates a buffer of at least $z_{1 - \epsilon}$ standard errors $1/\sqrt{l_{\bar{d}}''(\bs x^{(m_{\bar{d}})}|\hat{d}_{\bar{d}} )(n - m_{\bar{d}} - p - q)}$ between the estimate $\hat{d}_{\bar{d}}$ and the upper bound $\bar{d}$. When $\epsilon = 0.5$, this  ensures that $\hat{d}_{\bar{d}}< \bar{d}$, which is appropriate when approximating the sampling distribution of $\hat{d}_{\bar{d}}$ well is infeasible. We find that $\delta = 0.01$ performs well in practice; it is large enough to avoid numerical instability near the boundary and small enough to characterize the behavior of the log-likelihood at the boundary. We explore the choice of $\epsilon$ in simulations. 

\section{Simulations}

We use simulations to explore the performance of exact likelihood  and Whittle  likelihood estimators for fixed upper bounds $\bar{d} \in \left\{0.5, 1.5, 2.5, 3.5\right\}$ for $\mu_t = \mu$.  We focus on ARFIMA$\left(0, d, 0\right)$ models in simulations because computation for fitting ARFIMA$\left(p, d, q\right)$ models with $p > 0$ or $q > 0$ is much slower than computation for fitting ARFIMA$\left(0, d, 0\right)$ models. ARFIMA$\left(p, d, q\right)$ models are fit in two applications in Section~\ref{sec:app}.
We consider interval estimation and sampling distribution based adaptive estimators for exact likelihood based estimators only.  
We consider the same set of simulations performed in \cite{Beran1995} and \cite{Mayoral2007}. We examine the performance of alternative estimators of $d$ for true memory parameter values $d \in \left\{-0.7, -0.3, -0.2,  0.0, 0.2, 0.4, 0.7, 0.8, 1.0, 1.2, 1.4, 2.0, 2.2\right\}$ and sample sizes $n \in \left\{100, 200, 400, 500\right\}$. For each pair of $d$ and $n$ values, we simulate $100$ ARFIMA$\left(0, d, 0\right)$ time series, fixing $\mu = 0$ and $\sigma^2 = 1$. 
 For \color{black}$0.5 \leq d < 1.5$ and $1.5 \leq d < 2.5$\color{black}, simulated time series are obtained by taking \color{black} one \color{black} or two cumulative sums of simulated stationary ARFIMA time series, respectively. \color{black}

\subsection{Point Estimation of the Memory Parameter for Fixed $\bar{d}$}

\begin{figure}[ht!]
\centering
\includegraphics[scale = 1]{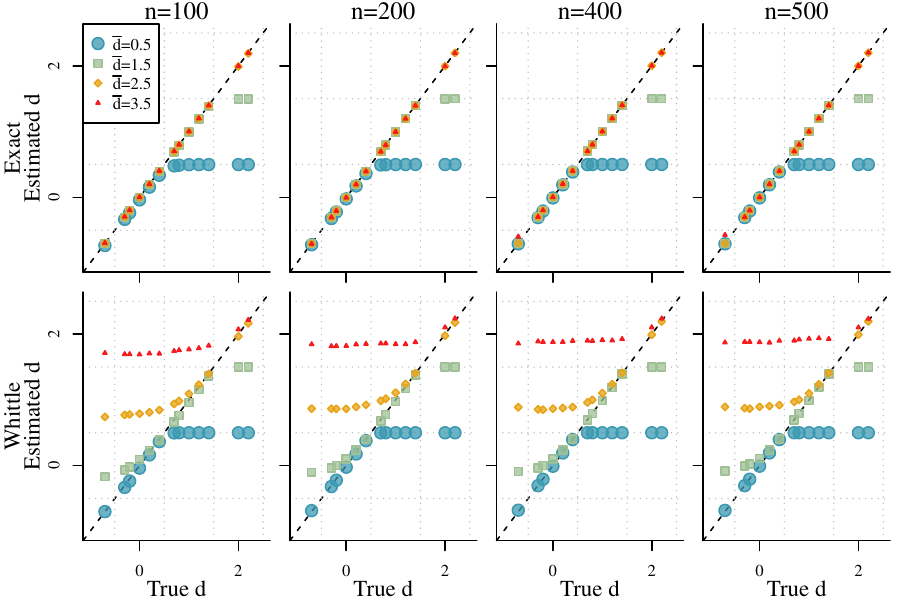}
\caption{Average estimates of $d$ across 100 ARFIMA$\left(0, d, 0\right)$ time series with $\mu = 0$ and $\sigma^2 = 1$ from  maximizing exact  or  Whittle  likelihoods  with respect to $d$, $\mu$, and $\sigma^2$.
}
\label{fig:mean_ests}
\end{figure}

 The first row of  Figure~\ref{fig:mean_ests} shows the average estimates of $d$ obtained by maximizing the exact likelihood  with respect to $d$, $\mu$, and $\sigma^2$. We consider upper bounds $\bar{d}$ that exceed the true data generating value of $d$.
We estimate $d$ well at all sample sizes as long as $\bar{d}$ exceeds the true value of $d$. Surprisingly, we estimate $d$ well even when the true value of $d$ is much smaller than $\bar{d}$, e.g., when the true value of $d$ is negative and $\bar{d} > 0.5$. This suggests that we do not pay a high price for overdifferencing the data when we are interested in obtaining a point estimate of the memory parameter $d$ by maximizing the exact likelihood.

The second row of Figure~\ref{fig:mean_ests} shows average estimates of $d$ obtained by maximizing the Whittle approximate likelihood with respect to $d$, $\mu$, and $\sigma^2$, respectively.   For the Whittle approximate likelihood estimator, we pay a price for overdifferencing. Whittle  estimates of $d$ are much more sensitive to the choice of $\bar{d}$; they  are biased if $\bar{d}$ is too large or too small. Not only do they underestimate the memory parameter $d$  when the true value of the memory parameter exceeds $\bar{d}$, they also systematically overestimate $d$ when $\bar{d}$ exceeds the true value of $d$, especially when $\bar{d}$ is much larger than the true value of $d$. 
They perform relatively well only when the true value $d$ satisfies $\bar{d} - 1.5 \leq d \leq \bar{d}$.  The Whittle estimator's  systematic overestimation of $d$ when the true value $d \leq -0.5$  is consistent with the literature. \cite{Hurvich1995} demonstrated that the log periodogram is biased and overestimates the memory parameter when the true value  $d \leq-0.5$, and \cite{Hurvich1998} observed this in simulations.

\subsection{Interval Estimation of the Memory Parameter for Fixed $\bar{d}$}

\begin{figure}[ht!]
\centering
\includegraphics[scale = 1]{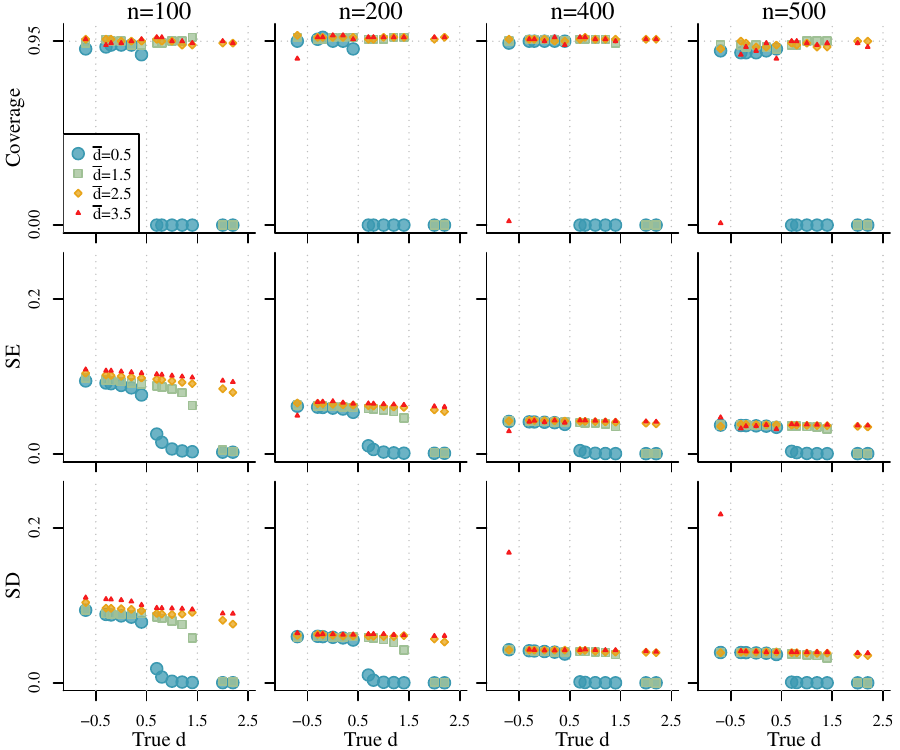}
\caption{Coverage of 95\% confidence intervals for $d$, average standard errors of $\hat{d}$, and standard deviations of $\hat{d}$  across 100 ARFIMA$\left(0, d, 0\right)$ time series with $\mu = 0$ and $\sigma^2 = 1$  from  maximizing the exact likelihood  with respect to $d$, $\mu$, and $\sigma^2$.} 
\label{fig:coverage}
\end{figure}

In practice, we may also be interested in the uncertainty of our estimate of the memory parameter $d$. For exact likelihood based estimators,  Figure~\ref{fig:coverage} shows the coverage of 95\% confidence intervals for the memory parameter $d$, average standard errors of $\hat{d}$, and approximate standard deviations of $\hat{d}$.   
Unsurprisingly, coverage is poor when the true value of the memory parameter $d$ exceeds the maximum value $\bar{d}$. Standard errors are larger and 95\% confidence intervals are wider  for larger values of $\bar{d}$, especially when $n$ is small.
Thus, we pay a small price for overdifferencing when we are interested in the uncertainty of our estimate of the memory parameter $d$ when $\bar{d}$, as long as $\bar{d}$ is not too large.
We pay a high price for overdifferencing if $\bar{d}$ is especially large relative to the true value of the long memory parameter, specifically if the difference between $\bar{d}$ exceeds the true value of the long memory parameter by three or more, even when $n$ is large.
Specifically, we obtain 95\% intervals with near zero coverage and standard errors that vastly underestimate the variability of $\hat{d}$. This is caused by numerical instability of the likelihood  
when $d \in \left(-0.5 + m_{\bar{d}} - j, 0.5 + m_{\bar{d}} - j\right)$ and $j \geq 4$. We evaluate the likelihood in this range by computing the likelihood of a stationary ARFIMA$\left(p, d - m_{\bar{d}} + j, j + q\right)$ model with $j + q$ constrained moving average parameters $\tilde{\theta}_1^{\left(j\right)},\dots, \tilde{\theta}^{\left(j\right)}_{j + q}$ obtained by expanding out $\left(1 - B\right)^j\theta\left(B\right)$.  This corresponds to the likelihood of a stationary process and thus corresponds to a positive definite variance-covariance matrix, however we find that this variance-covariance matrix can become  ill-conditioned when $n$ is large, which leads to failure of the adjusted version of Durbin's algorithm used to evaluate the likelihood. This is explored in detail in Section~\ref{appsec:sec9}  of the Appendix and is consistent with existing research that finds poor performance of the Durbin-Levinson algorithm for ill-conditioned positive definite Toeplitz matrices \citep{Gohbert1995}.

\subsection{Adaptive Point Estimation of the Memory Parameter}

\begin{figure}[ht!]
\centering
\includegraphics[scale = 1]{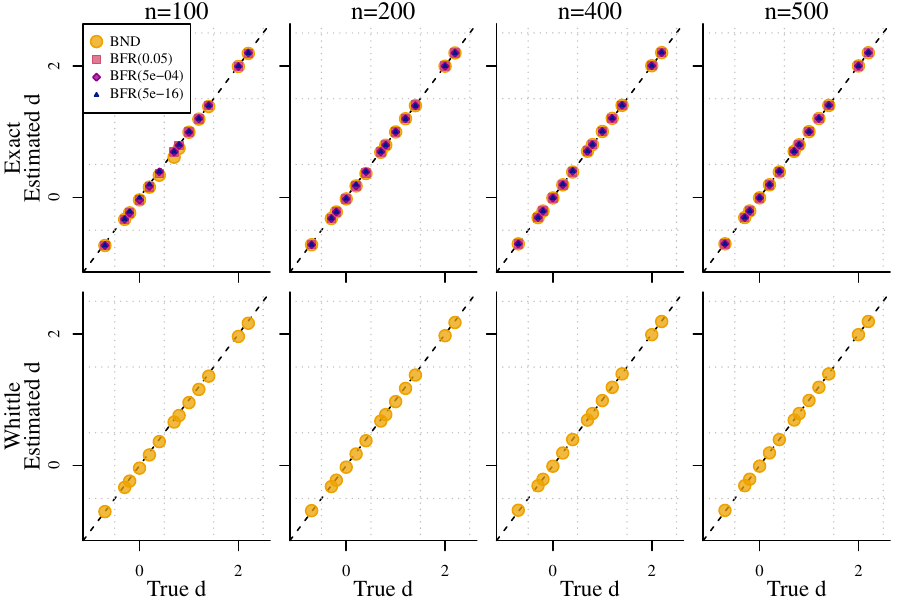}
\caption{Average BND and  BFR$\left(\epsilon\right)$ estimates of $d$ across $100$ ARFIMA$\left(0, d, 0\right)$ time series with $\mu = 0$ and $\sigma^2 = 1$ for different values of $\epsilon$, which determines the percentile used to choose $\bar{d}$. When estimating $d$, the $\mu$ and  $\sigma^2$ are treated as unknown. }
\label{fig:adests}
\end{figure}

The results of the previous section, in particular  the second row of Figure~\ref{fig:mean_ests} and the last two rows of Figure~\ref{fig:coverage},  motivate the need for the adaptive methods introduced in Section~\ref{subsec:adapt}.
We abbreviate the procedure obtained by setting $\epsilon = 0.5$  as  BND because it checks that the log-likelihood is decreasing in $d$ at the upper boundary $\bar{d}$. We  abbreviate procedures obtained by setting $\epsilon < 0.5$ as BFR$\left(\epsilon\right)$ because they ensure that a buffer between $\hat{d}$ and the upper boundary $\bar{d}$ by requiring that the approximate $\left(1 - \epsilon\right)\times100 \%$ percentile of $d$ be less than the upper boundary $\bar{d}$.  \color{black}The abbreviations BND and BFR are chosen to emphasize that the BND estimator prevents the estimator from being on the BouNDary and the BFR estimators ensure the presence of a BuFfeR between the estimators and the boundary. \color{black}
We consider three different values of $\epsilon = 5\times 10^{-s}$ corresponding to $s\in \left\{2, 4, 16\right\}$.
We note that $s = 16$ is the largest integer that returned a finite $\left(1 - 5\times 10^{-s}\right)\times100\%$ standard normal percentile, given the version of \texttt{R} we were using. The choice $s = 16$ requires that $\hat{d}$ be at least $8.014$ standard errors from the boundary $\bar{d}$. While this may seem extreme, it may be appropriate given that the standard error is an estimate of the standard deviation of the sampling distribution of $\hat{d}$ and given that a normal distribution may be a poor approximation of the sampling distribution of $\hat{d}$ in finite samples. 
Figure~\ref{fig:adests} shows average estimates of $d$ obtained using the BND and BFR$\left(\epsilon\right)$  exact and the BND Whittle  procedures. 
The BND  procedure  and all three BFR$\left(\epsilon\right)$  exact  procedures estimate $d$ well, regardless of the value of $\epsilon$ chosen or true value of the memory parameter. 
The BND Whittle procedure  produces excellent 
estimates of $d$ regardless of \color{black}its \color{black} true value, especially when $n$ is large.

\subsection{Adaptive Interval Estimation of the Memory Parameter}\label{subsec:adint}

\begin{figure}[ht!]
\centering
\includegraphics{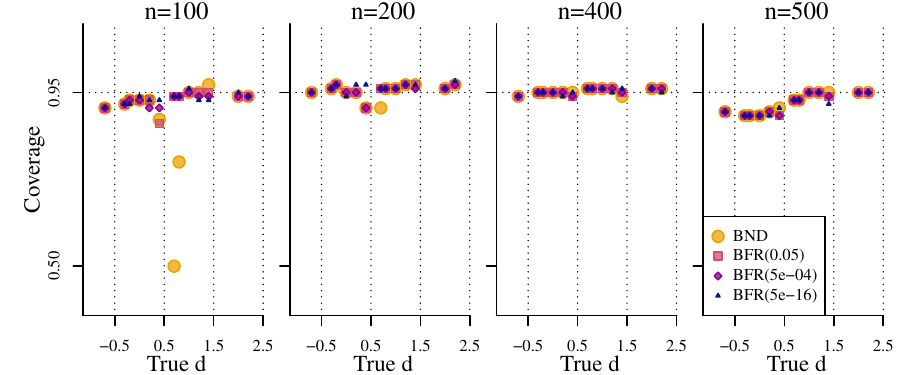}
\caption{Coverage of 95\% confidence intervals for  $d$ across 100 ARFIMA$\left(0, d, 0\right)$ time series with $\mu = 0$ and $\sigma^2 = 1$ for different values of $\epsilon$, which determines the percentile used to choose $\bar{d}$.  When estimating $d$, $\mu$ and $\sigma^2$ are treated as unknown. 
}
\label{fig:adcoverage}
\end{figure}

Last, we examine coverage of BND and BFR$\left(\epsilon\right)$ $95\%$ confidence intervals for $d$ obtained from the exact likelihood maximizing estimates.
The BND and BFR$\left(\epsilon\right)$ $95\%$ confidence intervals produce interval estimates whose coverage converges to $0.95$ as $n$ increases. 
BFR$\left(\epsilon\right)$ intervals improve most over BND intervals when the true memory parameter $d$ is near boundary values $0.5, 1.5, 2.5$ and $n$ is small. This is when BND estimates $\hat{d}_{\bar{d}}$ are more likely to be close to the boundary $\bar{d}$, where asymptotic normal approximations to the sampling distribution of $\hat{d}_{\bar{d}}$ are known to be poor.
Similar patterns are observed when comparing BFR$\left(\epsilon\right)$ intervals computed using smaller versus larger values of $\epsilon$. 
Additionally, the performance of BFR$\left(\epsilon\right)$ $95\%$ confidence intervals does not seem very sensitive to the choice of $\epsilon$, especially when $n$ is large. 
Figure~\ref{fig:adcoverage} shows that the best coverage rates were obtained by setting $\epsilon = 5\times 10^{-16}$.  This  suggests the simple strategy of choosing the smallest possible $\epsilon$.

\subsection{Comparison  to Alternatives}\label{subsec:comp}

We compare the absolute bias of the adaptive exact and Whittle  likelihood estimators  of $d$  to the absolute bias of two alternatives, one which treats the mean $\mu$ and variance $\sigma^2$ as unknown and one which treats the mean $\mu$ as known and equal to its true value and treats $\sigma^2$ as unknown. We compare to a CSS approximate likelihood estimator as described in \cite{Beran1995}, \cite{Hualde2011}, and \cite{Hualde2020} which treats the mean $\mu$ and variance $\sigma^2$ as unknown, and we compare to \color{black} what we call \color{black} the most favorable oracle implementation of the generalized minimum distance (GMD) estimator introduced by \cite{Mayoral2007}, which treats the mean $\mu$ as known and equal to its true value and treats $\sigma^2$ as unknown. The GMD estimator requires specification of the number of sample autocorrelations to include in the objective function, denoted by $k$. \color{black}We refer to our implementation as the most favorable oracle implementation because it chooses the value of sample autocorrelations $k$ that is most favorable, insofar as it yields the least biased estimator of $d$ for a given sample size $n$, as if it were an oracle, i.e. as if it knew the true bias associated with each value of $k$.  We do not consider the estimator introduced by \cite{Velasco2000} because simulation results from \cite{Mayoral2007} show that it is more biased and variable than the GMD estimator for these simulation settings. Figure~\ref{fig:biascomp} shows the average absolute bias of the BND exact estimator, the BFR$\left(5\times 10^{-16}\right)$ exact estimator and the BND Whittle estimator compared to the average absolute bias of the CSS estimator and the Best GMD estimator. Average absolute bias for the Best GMD estimator is reprinted from \cite{Mayoral2007}. As in this paper,  \cite{Mayoral2007} simulates ARFIMA$\left(p, d, q\right)$ time series with mean $0$ and error variance $\sigma^2 = 1$, but  performs 5,000 simulations  for each pair of memory parameter $d$ and sample size $n$ values. 

\begin{figure}[ht!]
\vspace{-7mm}
\centering
\includegraphics[scale = 1]{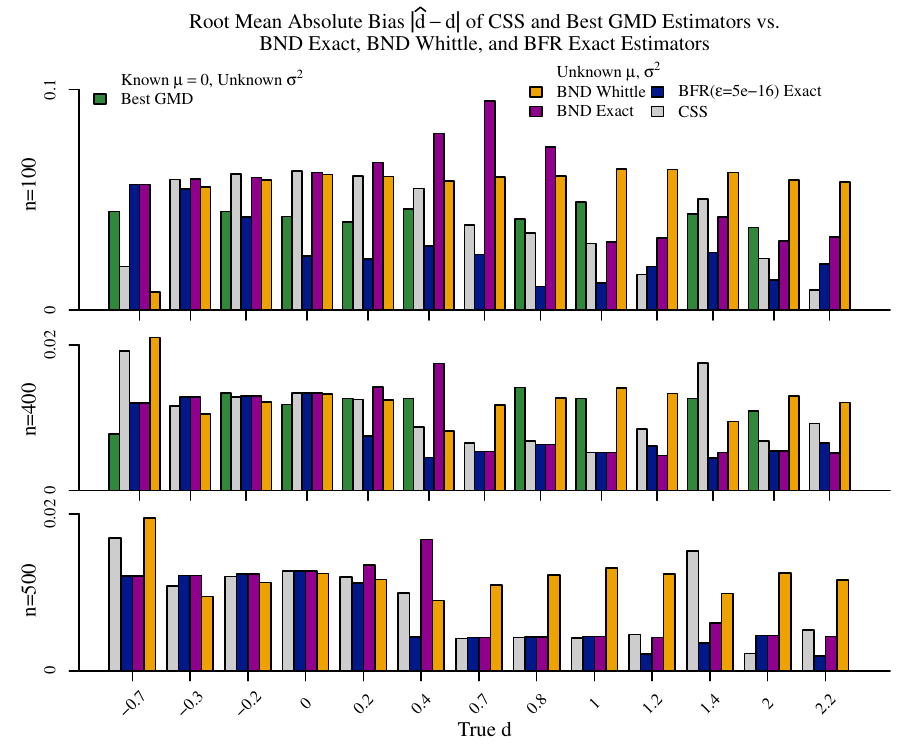} 
\caption{Root  \color{black}mean \color{black} absolute bias of CSS, BND, and BFR$\left(5\times 10^{-16}\right)$ estimators and the Best GMD estimator of $d$ across 100 and 5,000 ARFIMA$\left(0, d, 0\right)$ time series with mean $\mu = 0$ and variance $\sigma^2 = 1$, respectively, reprinted from \cite{Mayoral2007} for Best GMD. For CSS, BND, and BFR$\left(5\times 10^{-16}\right)$, we assume $\mu_t = \mu$ and treat $\mu$ and $\sigma^2$ as unknown.}
\label{fig:biascomp}
\end{figure}

Figure~\ref{fig:biascomp} shows that our adaptive exact estimators, especially the BFR$\left(5\times 10^{-16}\right)$ estimator, are less biased in many settings. Our adaptive Whittle likelihood estimator is the most biased estimator when $d \geq 0.5$. 
The BFR$\left(5\times 10^{-16}\right)$ outperforms the Best GMD estimator when $d \geq 0$ regardless of the sample size. As the sample size increases, the BND exact estimator outperforms the Best GMD estimator for all $d \geq 0$.
This is noteworthy given that the BND and BFR$\left(\epsilon\right)$ estimators summarized in Figure~\ref{fig:biascomp} treat the overall mean $\mu$  and variance $\sigma^2$  as unknown, whereas  the Best GMD estimator treats  the overall mean $\mu$ as constant and known to be equal to $0$  and the variance $\sigma^2$ as unknown.

Comparison to the CSS estimator suggests systematically poorer performance of the CSS estimator versus most adaptive estimators when the true value of $d$ is close to the boundaries of stationarity, $0.5$, and $1.5$, even as the sample size increases. 
When the sample size is smaller, the BFR$\left(5\times 10^{-16}\right)$ estimator performs better than the CSS estimator when $-0.3 \leq d \leq 1.4$, the BND exact estimator performs better  than the CSS estimator when $d \in \left\{-0.2, 0, 1.4\right\}$. As the sample size increases, the BFR$\left(5\times 10^{-16}\right)$ estimator outperforms the CSS estimator when $d \in \left\{-0.7, 0, 0.2, 0.4, 1.4\right\}$, the BND exact estimator outperforms the CSS estimator when $d \in \left\{-0.7, 0, 1.4\right\}$. 
We also find that the variability of our adaptive estimators is comparable to the variability of the CSS estimator and the Best GMD estimator for most sample sizes and true values of $d$.
Additionally, we find that our adaptive estimators of the noise variance $\sigma^2$, especially the BND exact and BFR$\left(5\times 10^{-16}\right)$ exact estimators of $\sigma^2$, tend to be less biased than the CSS estimator of $\sigma^2$ in the presence of long memory, even as the sample size increases. More details regarding variability of estimators and estimation of the noise variance are provided in Sections~\ref{appsec:sec3}  and \ref{appsec:sec4}  of the Appendix. We did not explore estimation of the overall mean $\mu$ because, like other estimators based on integer differenced data, the BND and BFR$\left(\epsilon\right)$ estimators of $\mu$ are not unique after differencing.  In additional simulations summarized in Sections~\ref{appsec:sec5}-\ref{appsec:sec7} of the Appendix, we find that the relative performances of the proposed estimators to each other and alternatives persist when the mean is assumed to be a quadratic time trend that is not cancelled by pre-differencing when $\bar{d} \leq 3.5$, and when observations are \color{black}  non-Gaussian, specifically when observations are simulated by replacing the mean zero variance $\sigma^2$ normal errors $z_t$  in \eqref{eq:arfima} with heavier-than-normal and lighter-than-normal tailed mean zero  variance $\sigma^2$ generalized normal distributed errors $z_t$ with generalized normal shape parameters $q = 1$ and $q = 6$ \citep{Griffin2018b}. The generalized normal distributions with shape parameter $q = 1$ and $q = 2$ are the Laplace and normal distributions, respectively\color{black}.

\section{Applications}\label{sec:app}

\subsection{Chemical Process Concentration and Temperature}

We begin by applying our methods to  fit ARFIMA$\left(0, d, 0\right)$ models with $\mu_t = \mu$ to  the examples discussed in \cite{Beran1995}, which considered the same problem as this paper: chemical process concentration readings (Series A) and  chemical process temperature readings (Series C). 
Descriptive plots of Series A and Series C are provided in Section~\ref{appsec:sec10}  of the Appendix \color{black} and a more detailed exploration of sharp jumps that are observed in the Series C data and their compatibility with the assumption of normal errors are provided in Sections~\ref{appsec:sec12} and \ref{appsec:sec13} of the Appendix. \color{black}
For both, we consider exact likelihood estimation of the memory parameter because we are interested in obtaining point and interval estimates. Based on the results of the simulation study which suggest that choosing the smallest possible $\epsilon$ yields the best performance, we compute BFR$\left(\epsilon =5\times 10^{-16}\right)$ exact likelihood estimates. Analogous results for approximate likelihoods and corresponding BND estimates are provided in Section~\ref{appsec:sec10}  of the Appendix.

\begin{table}[ht!]
\centering
\footnotesize
\begin{tabular}{ccc|rc|ccc|rc}
Series & $n$ & $\bar{d}$ & \multicolumn{1}{c}{$\hat{d}_{\bar{d}}$}  &95\% Interval for $d$ & Series & $n$ & $\bar{d}$ & \multicolumn{1}{c}{$\hat{d}_{\bar{d}}$}  & \multicolumn{1}{c}{95\% Interval for $d$}  \\ \hline 
\multirow{3}{*}{A} & \multirow{3}{*}{$197$} & $0.5$ & $0.400$ & $\left(0.304, 0.496\right)$  & \multirow{3}{*}{C} & \multirow{3}{*}{$226$} &  $0.5$ & $0.500$  & $-$\\ 
& & $1.5$ &\cellcolor{lightgray}$0.427$ & \cellcolor{lightgray}$\left(0.319, 0.534\right)$  & & & $1.5$ & $1.500$ &  $-$ \\
& & $2.5$ & $0.436$ &  $\left(0.326, 0.545\right)$ & & &$2.5$  & \cellcolor{lightgray} $1.788$ & \cellcolor{lightgray}$\left(1.659,    1.918\right)$ 
\\ \hline 
\end{tabular}
\caption{Estimates and intervals for $d$ for  Series A and C, treating $\mu$ and $\sigma^2$ as unknown. BFR$\left(\epsilon =5\times 10^{-16}\right)$ estimates in gray. Intervals \color{black}shown if \color{black} log-likelihood \color{black} is decreasing  \color{black} at $\bar{d}$.}
\label{tab:apppe}
\end{table}

Table~\ref{tab:apppe} shows exact estimates of $d$ for $\bar{d} \in \left\{0.5, 1.5, 2.5\right\}$ with BFR$\left(\epsilon\right)$ exact likelihood estimates highlighted in gray.  Exact profile log-likelihood curves for $\bar{d} \in \left\{0.5, 1.5, 2.5\right\}$, with $\mu$ and $\sigma^2$ profiled out are provided in Section~\ref{appsec:sec10} of the Appendix. The exact likelihood curves retain the same shape as $\bar{d}$ increases; the maximizing value of $d$ and the curvature about the maximizing value change little as $\bar{d}$ increases. 
 We compute BFR$\left(\epsilon\right)$ estimates setting $\epsilon =5\times 10^{-16}$ based on the simulation results discussed in Section~\ref{subsec:adint}.
From Table~\ref{tab:apppe}, we see that the choice of $\bar{d}$ affects the exact likelihood estimates of $d$ minimally and the 95\% intervals for the exact likelihood estimates of $d$   substantively. 
The 95\% intervals for $\bar{d} = 1.5$, the BFR$\left(\epsilon =5\times 10^{-16}\right)$ exact likelihood estimate of $d$, and $\bar{d} = 2.5$, for the Series A data contain $0.5$, which suggests that the Series A data may not be stationary. 
This is consistent with \cite{Beran1995}, which used  an  approximate likelihood method to obtain an estimate $\hat{d} = 0.41$ and a 95\% confidence interval of $\left(0.301, 0.519\right)$.

\begin{table}[ht!]
\centering
\footnotesize
\begin{tabular}{ccc|ccc}
& & & \multicolumn{3}{c}{95\% Interval} \\
Data & $n$ & $\bar{d}$ & $d$ & $\theta_1$ & $\phi_1$ \\ \hline \hline
\multirow{4}{*}{Series A} & \multirow{4}{*}{$197$} &0.5 & $\left(0.286, 0.553\right)$ & $\left(-0.227, 0.152\right)$ & $-$ \\ 
& & $1.5$ &\cellcolor{lightgray}$\left(0.298, 0.707\right)$ &\cellcolor{lightgray}$\left(-0.367, 0.134\right)$ &\cellcolor{lightgray}$-$  \\
& & $2.5$ & $\left(1.138, 1.490\right)$  & $\left(-1.014, -0.833\right)$ & $-$ \\
& & $3.5$ & $\left(1.125, 1.495\right)$  & $\left(-1.014, -0.809\right)$ & $-$ 
\\ \hline
\multirow{4}{*}{Series C} & \multirow{4}{*}{$226$} & $0.5$ & $-$ & $-$ & $-$  \\ 
& &$1.5$ &$\left(0.693, 1.208\right)$ &$-$ &$\left(0.682, 1.019\right)$ \\
& &$2.5$  &\cellcolor{lightgray}$\left(0.703, 1.241\right)$ &\cellcolor{lightgray}$-$ &\cellcolor{lightgray}$\left(0.662, 1.023\right)$ \\
& & $3.5$ & $\left(0.830, 1.113\right)$  &  $-$ & $\left(0.757, 0.946\right)$
\\ \hline
\end{tabular}
\caption{95\% intervals for exact likelihood estimates of the parameters of ARFIMA$\left(0, d, 1\right)$ and ARFIMA$\left(1, d, 0\right)$ processes  both with $\mu_t = \mu$  fit to the chemical process concentration readings (Series A) and chemical process temperature readings (Series C) for different values of $\bar{d}$. The BFR$\left(\epsilon =5\times 10^{-16}\right)$ exact likelihood estimates are highlighted in gray.  \color{black}Intervals shown if log-likelihood is decreasing  at $\bar{d}$.}
\label{tab:armaappci}
\end{table}

As in  \cite{Beran1995}, we also fit ARFIMA$\left(0, d, 1\right)$  and ARFIMA$\left(1, d, 0\right)$ models with $\mu_t = \mu$  to the Series A and C data, respectively,  for different values of $\bar{d}$. ARFIMA$\left(p, d, q\right)$  models  are more computationally challenging to fit than ARFIMA$\left(0, d, 0\right)$ models  because evaluating the likelihood is much computationally intensive and the likelihood can have many modes. 
Table~\ref{tab:armaapppe} shows the corresponding parameter estimates we obtain for both data sets for $\bar{d} \in \left\{0.5, 1.5, 2.5, 3.5\right\}$, with the BFR$\left(\epsilon =5\times 10^{-16}\right)$ exact likelihood estimates highlighted in gray. 
Examining the estimates of an ARFIMA$\left(0, d, 1\right)$ model for Series A, we observe striking changes in the exact likelihood estimates of $d$ and $\theta_1$ as $\bar{d}$ changes, and substantial differences between several of the exact and approximate maximum likelihood estimates. Furthermore, the 95\% intervals corresponding to the BFR$\left(\epsilon =5\times 10^{-16}\right)$ exact estimates for $d$ and $\theta_1$ shown in Table~\ref{tab:armaapppe} do not contain the exact likelihood estimates for $\bar{d} > 1.5$.
At the same time, the BFR$\left(\epsilon =5\times 10^{-16}\right)$ exact estimates for $d$ and $\theta_1$ align well with the estimates $\hat{d} = 0.445$, $\hat{\theta}_1 = -0.056$ and corresponding 95\% intervals $\left(0.261, 0.629\right)$ and $\left(-0.290, 0.179\right)$  found in \cite{Beran1995}.
This warrants more careful investigation.

\begin{table}[ht!]
\centering
\footnotesize
\begin{tabular}{cc|ccc|ccc}
 & & \multicolumn{3}{c|}{Estimate} & \multicolumn{3}{c}{95\% Interval} \\
Series & $\bar{d}$ & $d$ & $\theta_1$ & $\phi_1$  & $d$ & $\theta_1$ & $\phi_1$  \\ \hline 
\multirow{4}{*}{A} & $0.5$ & $0.419$ & $-0.037$ & $-$ & $\left(0.286, 0.553\right)$ & $\left(-0.227, 0.152\right)$ & $-$  \\ 
& $1.5$ 
&\cellcolor{lightgray}$0.502$  & \cellcolor{lightgray}$-0.117$ &  \cellcolor{lightgray}$-$ &\cellcolor{lightgray}$\left(0.298, 0.707\right)$ &\cellcolor{lightgray}$\left(-0.367, 0.134\right)$ &\cellcolor{lightgray}$-$\\
& $2.5$ & $1.314$  & $-0.923$ & $-$ & $\left(1.138, 1.490\right)$  & $\left(-1.014, -0.833\right)$ & $-$\\ 
& $3.5$ & $1.310$ & $-0.911$ & $-$   & $\left(1.125, 1.495\right)$  & $\left(-1.014, -0.809\right)$ & $-$ \\ \hline
\multirow{4}{*}{C} &  $0.5$ 
& $0.500$ & $-$& $1.000$ & $-$ & $-$ & $-$ \\ 
& $1.5$ & $0.950$ & $-$ & $0.850$ &$\left(0.693, 1.208\right)$ &$-$ &$\left(0.682, 1.019\right)$ \\
&$2.5$  & \cellcolor{lightgray}$0.972$ & \cellcolor{lightgray}$-$ &  \cellcolor{lightgray}$0.842$ &\cellcolor{lightgray}$\left(0.703, 1.241\right)$ &\cellcolor{lightgray}$-$ &\cellcolor{lightgray}$\left(0.662, 1.023\right)$  \\
& $3.5$  & $0.971$ &$-$ & $0.852$ & $\left(0.830, 1.113\right)$  &  $-$ & $\left(0.757, 0.946\right)$
\\ \hline
\end{tabular}
\caption{
Estimates  and intervals for  ARFIMA$\left(0, d, 1\right)$ and ARFIMA$\left(1, d, 0\right)$  models with $\mu_t = \mu$  for  Series A  and  C, respectively, treating $\mu$ and $\sigma^2$ as unknown. BFR$\left(\epsilon =5\times 10^{-16}\right)$ estimates given in gray. Intervals provided when the log-likelihood is decreasing at $\bar{d}$.
}
\label{tab:armaapppe}
\end{table}
\normalsize

\begin{figure}[ht!]
\centering
\includegraphics{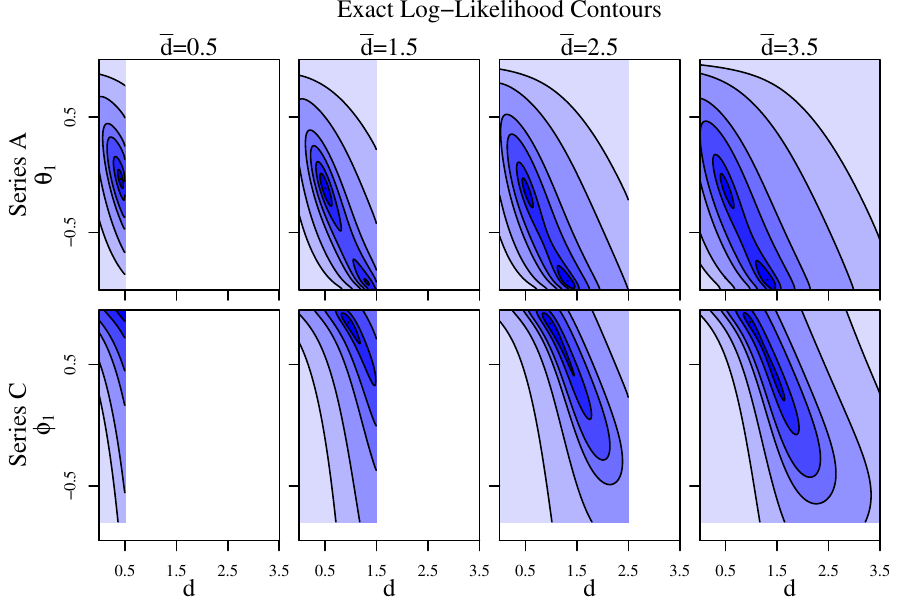}
\caption{Exact profile log-likelihoods for ARFIMA$\left(0, d, 1\right)$ and ARFIMA$\left(1, d, 0\right)$ models with $\mu_t = \mu$ for Series A and Series C data, respectively, with $\mu$ and $\sigma^2$ profiled out.}
\label{fig:seriesacar}
\end{figure}

The first row of  Figure~\ref{fig:seriesacar} shows the exact joint log-likelihoods for the Series A data as a function of the moving average parameter $\theta_1$ and $d$ for $\bar{d} \in \left\{0.5, 1.5, 2.5, 3.5\right\}$. \color{black}Estimates of $\theta_1$ and $d$ dramatically jump as $\bar{d}$ increases from 1.5 to 2.5.  We can see why by examining \eqref{eq:arfima}, specifically $\left(1 - B\right)^d  \left(y_t - \mu_{t} \right)= \left(1 + \theta_1 B\right)z_t$. For any value of $d \leq \bar{d} - 1$, the model obtained when  $\theta_1 = 0$ is equivalent to the model obtained when increasing the differencing parameter by 1 and setting $\theta_1 = -1$.  Thus, the \color{black} exact log-likelihood is multimodal when $\bar{d} > 0.5$, one with $d < 1$ and another with $d > 1$. Which mode maximizes the likelihood depends on the choice of $\bar{d}$; the maximum likelihood estimates for $\bar{d} \leq 1.5$ corresponds to the mode with $d < 1$, whereas the maximum likelihood estimate for $\bar{d} > 1.5$ corresponds to the mode with $d > 1$. This suggests possibly poor identifiability of the parameters of the ARIMA$\left(p, d, q\right)$ model even in simple cases with $p = 0$ and $q = 1$, which should be considered whenever an ARIMA$\left(p, d, q\right)$ is applied.

Table~\ref{tab:armaapppe} shows more stable estimates of the parameters of an ARFIMA$\left(1, d, 0\right)$ model for Series C   across values of $\bar{d}$. The BFR$\left(\epsilon =5\times 10^{-16}\right)$ exact likelihood estimates of $\hat{d}_{\bar{d}} = 0.972$ and $\hat{\phi}_{\bar{d},1} = 0.842$, with corresponding 95\% intervals of $\left(0.684, 1.261\right)$ and $\left(0.654, 1.031\right)$ are consistent with the approximate likelihood estimates $\hat{d} = 0.905$ and $\hat{\phi}_1 = 0.864$  and 95\% confidence intervals $\left(0.662, 1.148\right)$ and $\left(0.708, 1.00\right)$ provided in \cite{Beran1995}.
Although these estimates are more stable, the joint log-likelihoods in  the second row of  Figure~\ref{fig:seriesacar} are banana shaped near the maximum. Again, this reflects somewhat poor identifiability of $d$ and $\phi_1$ and further emphasizes the need for care when applying ARIMA$\left(p, d, q\right)$ models.
 
\subsection{$\text{CO}_2$ Emissions}

\cite{Barassi2018} uses long memory models to assess whether deviations of relative per capita CO$_2$ emissions of 28 OECD countries from country-specific linear time trends are mean reverting. 
Letting $y_{tc}$ be the relative per capita CO$_2$ emissions of country $c$ at time $t$ as defined in \cite{Barassi2018}, we assume that deviations of the relative per capita CO$_2$ emissions of each country $y_{tc}$ from a  country-specific linear time trend $ \mu_{tc} = \mu_c +  \beta_c t$ are an ARFIMA$\left(0, d_c, 0\right)$ process, with $\left(1 - B\right)^{d_c} \left(y_{tc}  - \mu_c - \beta_c t \right)=z_{tc}$, $z_{tc} \stackrel{i.i.d.}{\sim}\mathcal{N}\left(0, \sigma^2_c\right)$, and error variances $\sigma^2_c$. Under this model, country $c$'s relative per capita CO$_2$ emissions are mean reverting if $d_c < 1$. They assess whether each country's relative per capita CO$_2$ emissions are mean reverting by testing the null hypothesis $d_c \geq 1$ for each country. \color{black} Descriptive plots of the data are provided in Sections \ref{appsec:sec14} and \ref{appsec:sec15} of the Appendix.\color{black}

\begin{figure}[ht!]
\centering
\includegraphics{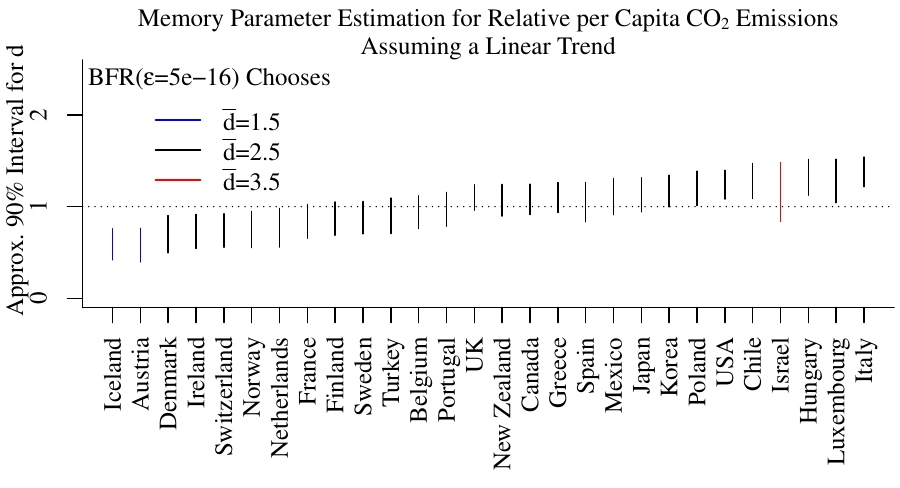}
\caption{90\% intervals for BFR$\left(\epsilon =5\times 10^{-16}\right)$ exact likelihood  estimates  of $d$.}
\label{fig:emissions}
\end{figure}

Treating the country-specific means $\mu_c$, slopes $\beta_c$, and variances $\sigma^2_c$ as unknown, we perform level $\alpha = 0.05$ tests of the null hypotheses that $d_c \geq 1$ versus the alternative that $d_c < 1$ for each country by comparing the upper bound of the 90\% confidence interval for the BFR$\left(\epsilon =5\times 10^{-16}\right)$ exact likelihood estimate $\hat{d}_c$ to $1$ and rejecting if it fails to exceed $1$. Figure~\ref{fig:emissions} shows the 90\% confidence intervals for the BFR$\left(\epsilon =5\times 10^{-16}\right)$ exact likelihood estimates for all 28 OECD countries. We reject the null hypothesis that $d_c \geq 1$ for Iceland, Austria, Denmark, Ireland, Switzerland, Norway, and the Netherlands.  This is largely consistent with  \cite{Barassi2018}, which reports strong evidence of mean reversion for Austria, Denmark, Finland, Iceland, Ireland, Israel, Norway and Switzerland based on a battery of alternative methods, many of which depend on the choice of several tuning parameters.
Unlike \cite{Barassi2018}, which finds strong evidence that Israel's per capita CO$_2$ emissions are mean reverting, we are not able to estimate the memory parameter for Israel with enough precision to reject the null hypothesis that  $d_c \geq 1$.

\subsection{ECIS Measurements}

Last, we compute BFR$\left(\epsilon =5\times 10^{-16}\right)$ exact likelihood estimates of the memory parameter for time series of ECIS measurements. ECIS monitors the growth and behavior of cells in culture \citep{Giaever1991}.
Previous research suggests that ARFIMA$\left(0, d, 0\right)$ models may be appropriate for ECIS measurements and hypothesizes that the memory parameter $d$ may vary by cell type and contamination status \citep{Giaever1991,Tarantola2010,Gelsinger2019, Zhang2020}. \color{black} As in \cite{Zhang2020}, w\color{black}e consider ECIS measurements from eight experiments, four measuring Madin-Darby canine kidney (MDCK) cells and four measuring African green monkey kidney epithelial (BSC-1) cells. \color{black}Plots of selected time series are provided in Section~\ref{appsec:sec16} of the Appendix. \color{black} From each experiment, we examine 40 individual time series of about $n \approx 170$ measurements, corresponding to measurements from 40 cell filled wells on a single tray collected between 40 and 72 hours after the wells were filled and the experiment began. Of the 40 wells in each experiment, half were prepared using one medium, BSA, and half were prepared using another medium, gel. Of the 20 wells prepared with the same medium, 12 contain cells contaminated by mycoplasma and 8 contain uncontaminated cells.
Letting $y_{temwf}$ refer to ECIS measurements for well $w$ prepared with medium $m$ in experiment $e$ measured at frequency $f$, we assume that  the mean $\mu_{temwf} = \mu_{emwf}$ is constant and 
\begin{align*}
\begin{array}{ll}
\left(1 - B\right)^{d_{emf}^{\left(1\right)}} \left(y_{temwf}  - \mu_{emwf}\right)=z_{temwf} & \text{ for wells $w$ with contaminated cells and} \\
\left(1 - B\right)^{d_{emf}^{\left(0\right)}} \left(y_{temwf}  - \mu_{emwf}\right)=z_{temwf} & \text{ otherwise,}
\end{array}
\end{align*}
where $z_{temwf} \stackrel{i.i.d.}{\sim}\mathcal{N}\left(0, \sigma^2_{emwf}\right)$.
We obtain pairs of estimates of the memory parameters $d^{\left(0\right)}_{emf}$ and $d^{\left(1\right)}_{emf}$ for each experiment, medium, and frequency.

\begin{sidewaysfigure}
\centering
\includegraphics[scale = 1]{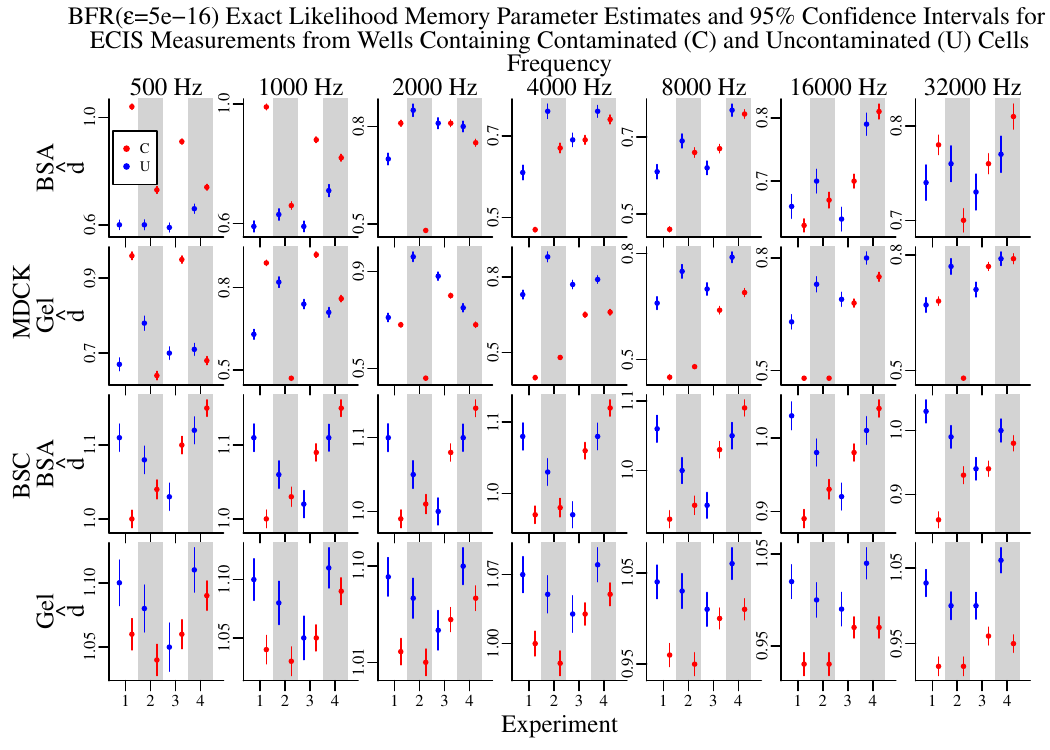}
\caption{BFR$\left(\epsilon =5\times 10^{-16}\right)$ exact likelihood estimates and  95\% intervals for  memory parameters $d^{\left(0\right)}_{emf}$ (blue) and $d^{\left(1\right)}_{emf}$ (red).  When estimating  $d^{\left(0\right)}_{emf}$ and $d^{\left(1\right)}_{emf}$, $\mu_{emwf}$ and $\sigma^2_{emwf}$ are treated as unknown. }
\label{fig:cells}
\end{sidewaysfigure}

Figure~\ref{fig:cells} suggests that ECIS measurements' long memory behavior may vary with contamination status. We see longer memory in measurements of contaminated MDCK cells using BSA at 32,000  Hz and shorter memory in measurements of contaminated MDCK cells using gel at 4,000 and 8,000 Hz, contaminated BSC-1 cells using BSA at  32,000  Hz, and contaminated BSC-1 cells using gel at all frequencies, especially 32,000 Hz. This suggests that long memory behavior of ECIS measurements may help distinguish contaminated from uncontaminated cells, although the sample sizes are small and the scope is limited.

\section{Conclusion}

We make a simple but powerful observation that allows us to  perform exact likelihood estimation of  the parameters of ARFIMA$\left(p, d, q\right)$ models without restricting the memory parameter $d$ to the range of values that correspond to a stationary ARFIMA$\left(p, d, q\right)$ process. We introduce adaptive procedures  for specifying an upper bound for the memory parameter   and demonstrate the utility of these procedures via simulation studies and applications, both to canonical datasets that have been explored in the  ARFIMA$\left(p, d, q\right)$ literature and to more modern datasets.  This allows us to identify situations where approximate likelihood methods perform well relative to exact likelihood methods and where they do not.

There are  many  future directions to pursue. Simulation studies investigating the finite sample properties of our proposed estimators when $p + q > 0$ would be valuable for assessing the relative merits of the proposed approach compared to alternatives. 
Alternative approaches to approximating standard errors and confidence intervals may provide better standard error and interval estimates.
We have used numerical differentiation to approximate standard errors and confidence intervals. As shown in Section~\ref{appsec:sec8}  of the Appendix, confidence intervals corresponding to BND and BFR$\left(\epsilon\right)$ exact likelihood estimators perform well on average across the simulations considered in this paper. However, they can be unstable in individual simulations.
Asymptotic or bootstrap-based approximate standard errors and confidence intervals may perform better.  
Replacing the adjusted version of Durbin's algorithm used to evaluate the likelihood with an alternative algorithm that  is more suitable for ill conditioned covariance matrices, e.g. the algorithm described in \citep{Gohbert1995},  may improve performance of the exact likelihood estimator for large maximum values $\bar{d}$ when the true value of the  memory  parameter $d$ is much smaller than $\bar{d}$, thus reducing the need for an adaptive approach altogether. The performance of our proposed estimators for non-Gaussian processes in simulations \color{black} and the evidence of non-Gaussian errors for the time series data considered in this paper \color{black} suggests the value of extending the consistency and asymptotic normality results provided in \cite{Lieberman2012} to non-Gaussian processes \color{black} and the value of generalizing our proposed method to allow for certain types of non-Gaussian errors\color{black}. Incorporating approximate sampling distributions  for  Whittle estimators, e.g. those described in \cite{Velasco2000},  may yield BFR$\left(\epsilon\right)$ Whittle estimators that outperform the BND Whittle estimators considered here. Alternative likelihood approximations may improve estimation of the memory parameter. Using one of the novel approximations introduced in \cite{Jesus2017}, \cite{Sykulski2019}, or \cite{Das2020} in place of the Whittle likelihood,  which is known to produce biased estimates \citep{Contreras-Cristan2006}, may yield better approximate estimators. Our approach may also provide a useful framework for exact likelihood estimation of models that allow the differencing parameter to vary over time and multivariate fractional differencing models as considered in \cite{Graves2015} and  \cite{Nielsen2015}, respectively.

\bigskip
\begin{center}
{\large\bf SUPPLEMENTARY MATERIAL}
\end{center}

A stand-alone package for implementing the methods described in this paper can be downloaded from \url{https://github.com/maryclare/nslm}.

\bibliographystyle{chicago}

\bibliography{LongMemory}

\begin{appendix}

\section{Continuity of $l_{\bar{d}}\left(\bs x^{\left(m_{\bar{d}}\right)}| d, \mu_{t}, \sigma, \bs \theta, \bs \phi \right)$}\label{appsec:sec1}

\begin{theorem}
The log-likelihood $l_{\bar{d}}\left(\bs x^{\left(m_{\bar{d}}\right)}| d, \mu_{t}, \sigma, \bs \theta, \bs \phi \right)$ is continuous for $-0.5 + m_{\bar{d}} - k\leq d < \bar{d}$. 
\end{theorem}

\begin{proof} 
Continuity of $l_{\bar{d}}\left(\bs x^{\left(m_{\bar{d}}\right)}| d, \mu_{t}, \sigma, \bs \theta, \bs \phi \right)$ on the intervals $\left[-0.5 + m_{\bar{d}} - j, 0.5 + m_{\bar{d}} - j \right)$ for $j = 0, \dots, k$ follows from continuity of the log-likelihood of a stationary ARFIMA$\left(p, d, q\right)$ process \citep{Dahlhaus1989}. Note that \cite{Dahlhaus1989} requires stationarity of the ARFIMA$\left(p, d, q\right)$ but does not require that roots of the moving average polynomial $\theta\left(x\right)$ lie strictly outside the unit circle.

Continuity of $l_{\bar{d}}\left(\bs x^{\left(m_{\bar{d}}\right)}| d, \mu_{t}, \sigma, \bs \theta, \bs \phi \right)$ at $d = -0.5 + m_{\bar{d}} - j$ for $j = 0, \dots, k-1$ requires
\begin{align*}
&\text{lim}_{\varepsilon \rightarrow 0^+} l_{\bar{d}}\left(\bs x^{\left(m_{\bar{d}}\right)}| 0.5 + m_{\bar{d}} - \left(j+1\right) - \varepsilon, \mu_{t}, \sigma, \bs \theta, \bs \phi \right) = \\ &\quad \quad \quad \quad \quad \quad l_{\bar{d}}\left(\bs x^{\left(m_{\bar{d}}\right)}| -0.5 + m_{\bar{d}} - j, \mu_{t}, \sigma, \bs \theta, \bs \phi \right).
\end{align*} 
Let $\boldsymbol \Omega\left(d, \bs \theta, \bs \phi, \sigma \right)$ refer to the $\left(n-m_{\bar{d}}\right)\times \left(n-m_{\bar{d}}\right)$ covariance matrix of the stationary differenced data $\boldsymbol x^{\left(m_{\bar{d}}\right)}$ with elements $\omega\left(d, \bs \theta, \bs \phi, \sigma \right)_{ii'}$ given by the autocovariance function  $\omega\left(\left|i-i'\right|; d, \bs \theta, \bs \phi, \sigma \right)$. The log-likelihood $l_{\bar{d}}\left(\bs x^{\left(m_{\bar{d}}\right)}| d, \mu_{t}, \sigma, \bs \theta, \bs \phi \right)$ is
\begin{align*}
l_{\bar{d}}\left(\bs x^{\left(m_{\bar{d}}\right)}| d, \mu, \sigma, \bs \theta, \bs \phi \right) = -&n\text{log}\left(2\pi\right)/2 - \text{log}\left(\left|\bs \Omega\left(d, \bs \theta,\bs \phi, \sigma\right)\right|\right)/2 - \\
&\left(\bs x^{\left(m_{\bar{d}}\right)} - \bs \mu^{\left(m_{\bar{d}}\right)} \right)'\bs \Omega\left(d, \bs \theta,\bs \phi, \sigma\right)^{-1}\left(\bs x^{\left(m_{\bar{d}}\right)} - \bs \mu^{\left(m_{\bar{d}}\right)} \right)/2.
\end{align*}
Because the log-likelihood depends on $d$ only through the autocovariance function and is a continuous function of the autocovariance function, the log-likelihood $l_{\bar{d}}\left(\bs x^{\left(m_{\bar{d}}\right)}| d, \mu_{t}, \sigma, \bs \theta, \bs \phi \right)$ is continuous at $d = -0.5 + m_{\bar{d}} - j$ if the autocovariance function is continuous at $d = -0.5 + m_{\bar{d}} - j$, 
\begin{align*}
\text{lim}_{\varepsilon\rightarrow 0^+} \omega\left(h; 0.5 + m_{\bar{d}} - \left(j+1\right) - \varepsilon, \bs \theta,\bs \phi, \sigma\right) = \omega\left(h; -0.5 + m_{\bar{d}} - j, \bs \theta,\bs \phi, \sigma\right).
\end{align*} 
Letting $\gamma\left(h; d, \bs \theta, \bs \phi, \sigma\right)$ refer to the autocovariance function of a mean-zero stationary ARFIMA$\left(p, d, q\right)$ process with parameters $d$, $\bs \theta$, $\bs \phi$, and $\sigma^2$, and letting $\tilde{\bs \theta}^{(j)}$ refer to the coefficients of the moving average polynomial $\left(1 - B\right)^j\theta\left(B\right)$, we have
\footnotesize
\begin{align*}
\text{lim}_{\varepsilon\rightarrow 0^+}&\omega\left(h; 0.5 + m_{\bar{d}} - \left(j+1\right) - \varepsilon, \bs \theta,\bs \phi, \sigma\right) = \text{lim}_{\varepsilon\rightarrow 0^+} \gamma\left(h; 0.5 - \varepsilon, \tilde{\bs \theta}^{\left(j+1\right)},\bs \phi, \sigma\right) =\\
&\text{lim}_{\varepsilon\rightarrow 0^+}  \int_{-\pi}^\pi \text{exp}\left\{ih\nu \right\}\left(\frac{\sigma^2\left| \tilde{\theta}^{\left(j+1\right)}\left(\text{exp}\left\{-i\nu\right\}\right)\right|^2}{2\pi\left|\phi\left(\text{exp}\left\{-i\nu\right\} \right)\right|^2}\right)\left|1 - \text{exp}\left\{-i\nu\right\} \right|^{-2\left(0.5 - \varepsilon\right)}d\nu = \\
&\text{lim}_{\varepsilon\rightarrow 0^+}  \int_{-\pi}^\pi \text{exp}\left\{ih\nu \right\}\left(\frac{\sigma^2\left| \left(1 - \text{exp}\left\{-i\nu\right\} \right)^{j + 1}\theta\left(\text{exp}\left\{-i\nu\right\}\right)\right|^2}{2\pi\left|\phi\left(\text{exp}\left\{-i\nu\right\} \right)\right|^2}\right)\left|1 - \text{exp}\left\{-i\nu\right\} \right|^{-2\left(0.5 - \varepsilon\right)}d\nu = \\
&\text{lim}_{\varepsilon\rightarrow 0^+}  \int_{-\pi}^\pi \text{exp}\left\{ih\nu \right\}\left(\frac{\sigma^2\left| \left(1 - \text{exp}\left\{-i\nu\right\} \right)^{j}\theta\left(\text{exp}\left\{-i\nu\right\}\right)\right|^2}{2\pi\left|\phi\left(\text{exp}\left\{-i\nu\right\} \right)\right|^2}\right)\left|1 - \text{exp}\left\{-i\nu\right\} \right|^{-2\left(0.5 - \varepsilon - 1\right)}d\nu = \\
&\text{lim}_{\varepsilon\rightarrow 0^+}  \int_{-\pi}^\pi \text{exp}\left\{ih\nu \right\}\left(\frac{\sigma^2\left|\tilde{\theta}^{(j)}\left(\text{exp}\left\{-i\nu\right\}\right)\right|^2}{2\pi\left|\phi\left(\text{exp}\left\{-i\nu\right\} \right)\right|^2}\right)\left|1 - \text{exp}\left\{-i\nu\right\} \right|^{-2\left(-0.5 - \varepsilon\right)}d\nu = \\
&\text{lim}_{\varepsilon\rightarrow 0^+}   \gamma\left(h; -0.5 - \varepsilon, \tilde{\bs \theta}^{\left(j\right)},\bs \phi, \sigma\right)=  \gamma\left(h; -0.5, \tilde{\bs \theta}^{\left(j\right)},\bs \phi, \sigma\right)=\omega\left(h; -0.5 + m_{\bar{d}} - j, \bs \theta,\bs \phi, \sigma\right).
\end{align*} 
\normalsize

Throughout, we make use of derivations of the spectral density and autocovariance function of an ARFIMA$\left(p, d, q\right)$ process in \cite{Sowell1992a}. Although  \cite{Sowell1992a} focuses on stationary ARFIMA$\left(p, d, q\right)$ process with roots of the autoregressive and moving average polynomials $\phi\left(x\right)$ and $\theta\left(x\right)$ outside the unit circle and $-0.5 < d < 0.5$, the derivations themselves do not require that the roots of the moving average polynomial $\theta\left(x\right)$ lie strictly outside the unit circle and allow for $-1 < d \leq -0.5$. 
\end{proof}

\clearpage

\section{The SCSS Approximate Likelihood}\label{appsec:sec2}

The SCSS approximation is obtained by assuming that the finite differences 
\begin{align*}
\left(1 - B\right)_+^{d - m_{\bar{d}} + j}(x^{\left(m_{\bar{d}}\right)}_t - \mu^{\left(m_{\bar{d}}\right)}_t)
\end{align*} are distributed according to an ARFIMA$\left(p, 0, j + q\right)$ process. The SCSS likelihood is not continuous. The SCSS likelihood is not generally equivalent to the CSS likelihood, which assumes that the finite differences $\left(1 - B\right)_+^{d}\left(y_t - \mu_t\right)$ are distributed according to an ARFIMA$\left(p, 0, q\right)$ process.
 The second and third rows of Figure~\ref{fig:mean_ests} show average estimates of $d$ obtained by maximizing the Whittle and SCSS approximate likelihoods with respect to $d$, $\mu$, and $\sigma^2$, respectively.  For the SCSS estimator, we pay a price for overdifferencing. The SCSS approximate likelihood based estimators of $d$ performs relatively well only when the true value $d$ satisfies $\bar{d} - 1.5 \leq d \leq \bar{d}$.

\begin{figure}
\centering
\includegraphics[scale = 1]{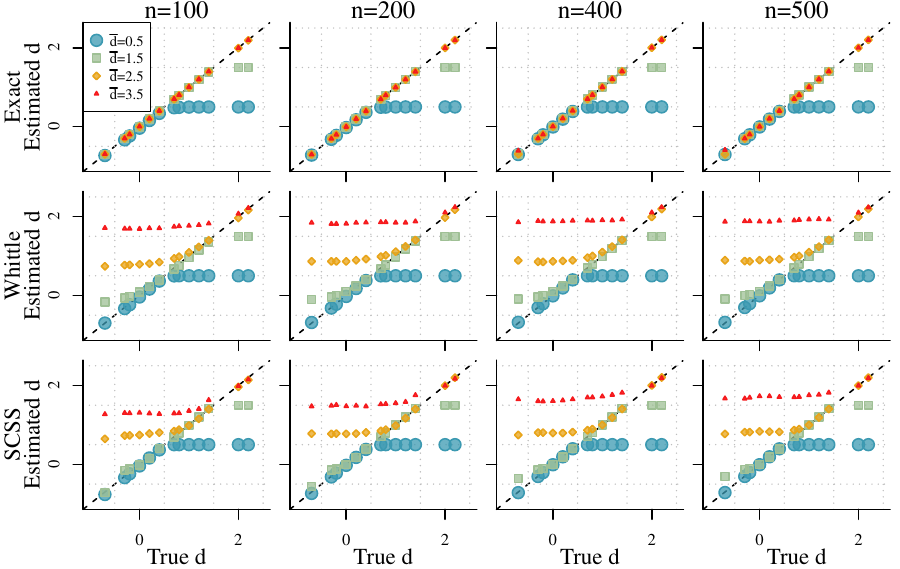}
\caption{Average estimates of $d$ across 100 ARFIMA$\left(0, d, 0\right)$ time series with $\mu = 0$ and $\sigma^2 = 1$ from  maximizing exact Whittle, or SCSS  likelihoods  with respect to $d$, $\mu$, and $\sigma^2$.
}
\label{fig:mean_ests}
\end{figure}

 At first glance, the SCSS estimator's tendency to overestimate $d$ when the true value is far less than  $\bar{d}$ may appear to conflict with the existing literature, specifically \cite{Hualde2011}. \cite{Hualde2011} demonstrate consistency and asymptotic normality of CSS estimators for arbitrary values of the differencing parameter $d$ when the truncated fractional differences $\left(1 - B\right)^d_+ y_t$ are  distributed according to a stationary ARMA$\left(p, q\right)$ model. However, our results depict the performance of the SCSS estimator when the untruncated fractional differences $\left(1 - B\right)^d y_t$ are distributed according to a stationary ARMA$\left(p, q\right)$ model. The truncated and untruncated models differ, especially when $d \leq -0.5$ and moreso as $d$ decreases below $-0.5$, which corresponds to the case where the true value of the differencing parameter $d$ is far less than the upper bound, specifically  $d \leq \bar{d} - 1$.
The differences between the truncated and untruncated models are illustrated in Section~\ref{appsec:sec11} of the Appendix.  We believe that this explains how we can observe the systematic overestimation of $d$ by the SCSS estimator shown in Figure~\ref{fig:mean_ests} while the results of \cite{Hualde2011} hold; the results of \cite{Hualde2011} hold under a different model for the observed time series data, which diverges more from the model we consider as the true value of $d$ gets further from the upper bound $\bar{d}$.  

\begin{figure}[h!]
\centering
\includegraphics[scale = 1]{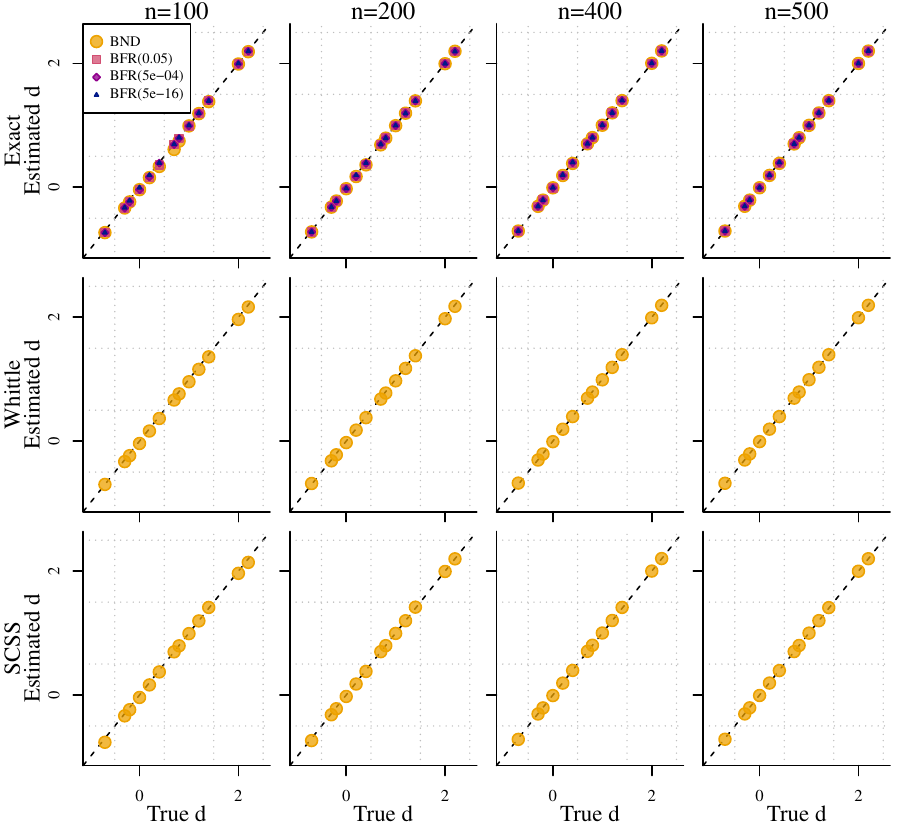}
\caption{Average BND and  BFR$\left(\epsilon\right)$ estimates of $d$ across $100$ ARFIMA$\left(0, d, 0\right)$ time series with $\mu = 0$ and $\sigma^2 = 1$ for different values of $\epsilon$, which determines the percentile used to choose $\bar{d}$. When estimating $d$, the $\mu$ and  $\sigma^2$ are treated as unknown. }
\label{fig:adests}
\end{figure}

Figure~\ref{fig:adests}  also  shows the average estimates of $d$ obtained using the BND  Whittle and SCSS  procedures, and indicates that the BND SCSS estimator also produces excellent estimates of $d$ 
regardless of the true value of the memory parameter, especially when $n$ is large. 

\begin{figure}
\centering
\includegraphics[scale = 1]{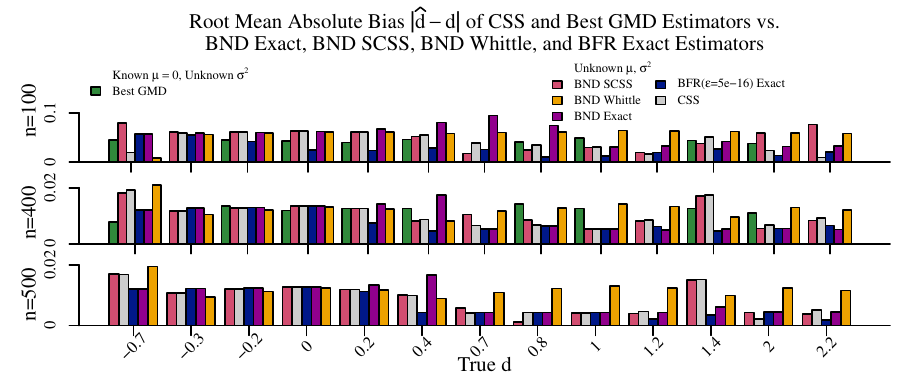} 
\caption{Root  average absolute bias of CSS, BND, and BFR$\left(5\times 10^{-16}\right)$ estimators and the Best GMD estimator of $d$ across 100 and 5,000 ARFIMA$\left(0, d, 0\right)$ time series with mean $\mu = 0$ and variance $\sigma^2 = 1$, respectively, reprinted from \cite{Mayoral2007} for Best GMD. For the CSS, BND, and BFR$\left(5\times 10^{-16}\right)$ estimators, we assume $\mu_t = \mu$ and treat $\mu$ and $\sigma^2$ as unknown.}
\label{fig:biascomp}
\end{figure}

We add a comparison of the absolute bias of the BND SCSS estimator to the adaptive exact and Whittle  likelihood estimators  of $d$  to the absolute bias of two alternatives depicted in Figure~7 of the main manuscript.
Figure~\ref{fig:biascomp} shows that our adaptive exact and SCSS estimators, especially the BFR$\left(5\times 10^{-16}\right)$ estimator, are less biased in many settings, whereas our adaptive Whittle likelihood estimator is the most biased estimator when $d \geq 0.5$ and performs similarly to our adaptive SCSS estimator otherwise. 
In comparison to the Best GMD estimator, the BFR$\left(5\times 10^{-16}\right)$ performs better when $d \geq 0$ regardless of the sample size. As the sample size increases, the BND exact estimator outperforms the Best GMD estimator for all $d \geq 0$ and the BND SCSS estimator outperforms the Best GMD estimator for most $d \geq 0$.
This is especially noteworthy given that the BFR$\left(\epsilon\right)$ estimator summarized in Figure~\ref{fig:biascomp} treats the overall mean $\mu$  and variance $\sigma^2$  as unknown, whereas  the Best GMD estimator treats  the overall mean $\mu$ as constant and known to be equal to $0$  and the variance $\sigma^2$ as unknown.

Comparison to the CSS estimator suggests systematically poorer performance of the CSS estimator versus most adaptive estimators when the true value of $d$ is close to the boundaries of stationarity, $0.5$, and $1.5$, even as the sample size increases. 
When the sample size is smaller, the BFR$\left(5\times 10^{-16}\right)$ estimator performs better than the CSS estimator when $-0.3 \leq d \leq 1.4$, the BND exact estimator performs better  than the CSS estimator when $d \in \left\{-0.2, 0, 1.4\right\}$, and the BND SCSS estimator performs better than the than the CSS estimator  when $0 \leq d \leq 1$ and $d = 1.4$. As the sample size increases, the BFR$\left(5\times 10^{-16}\right)$ estimator outperforms the CSS estimator when $d \in \left\{-0.7, 0, 0.2, 0.4, 1.4\right\}$, the BND exact estimator outperforms the CSS estimator when $d \in \left\{-0.7, 0, 1.4\right\}$, and the BND SCSS estimator outperforms the CSS estimator when $-0.7 \leq d \leq -0.2$, $0.2 \leq d \leq 1$, and $d = 1.4$.

\section{Variability Comparison with Alternatives}\label{appsec:sec3}

\begin{figure}
\centering
\includegraphics{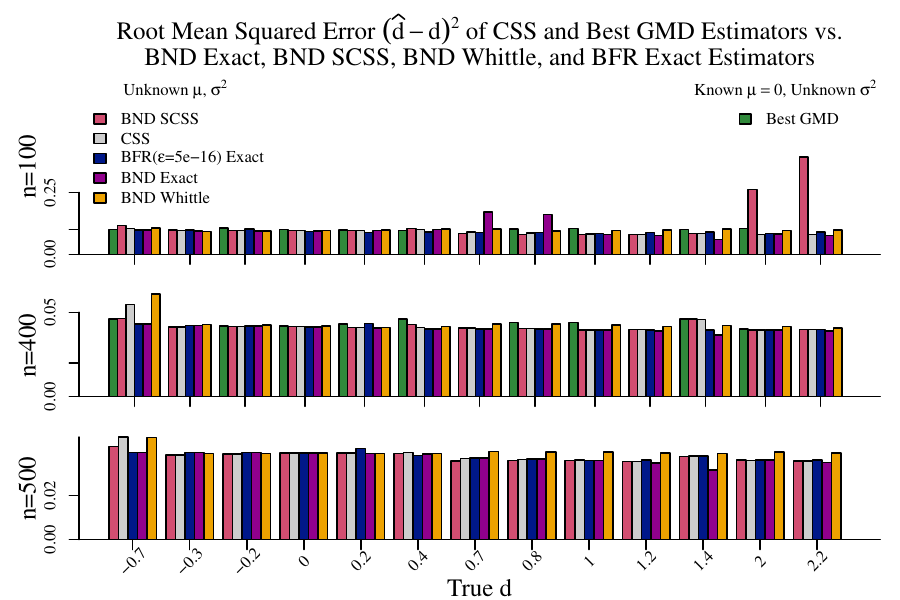}
\caption{For CSS, BND, and BFR$\left(\epsilon = 5\times 10^{-16}\right)$, root mean squared error (RMSE) of estimates of the memory parameter $\hat{d}$ is approximated from $100$ simulated ARFIMA$\left(0, d, 0\right)$ time series with mean $\mu = 0$ and variance $\sigma^2 = 1$ for each sample size $n$ and true value of the differencing parameter $d$. Root mean squared error (RMSE) for the best GMD estimator is reprinted from \cite{Mayoral2007} by choosing the smallest average RMSE across GMD estimators that use a different number of autocorrelations at each value of the true memory parameter $d$ and sample size $n$.}
\label{fig:rmsecomp}
\end{figure}

Figure~\ref{fig:rmsecomp} shows that the variability of our adaptive estimators is comparable to the variability of the CSS estimator and the ``Best GMD'' estimator for most sample sizes and true values of $d$. Again, this is especially noteworthy given that the adaptive BND SCSS, BND exact, and BFR$\left(\epsilon\right)$ exact estimators summarized in Figure~\ref{fig:rmsecomp} treat the overall mean $\mu$ and variance $\sigma^2$ as unknown, whereas the version of Mayoral's estimator summarized in Figure~\ref{fig:rmsecomp} treats the overall mean $\mu$ as constant and known to be equal to $0$ and the variance $\sigma^2$ as unknown.
There are some exceptions when the sample size is smaller; the adaptive BND SCSS estimator is especially variable when $d \geq 2$ and the adaptive BND exact estimator is especially variable when $0.5 < d < 1$.

\clearpage

\section{Estimation of $\sigma^2$ Comparison with Alternatives}\label{appsec:sec4}

\begin{figure}[h]
\centering
\includegraphics{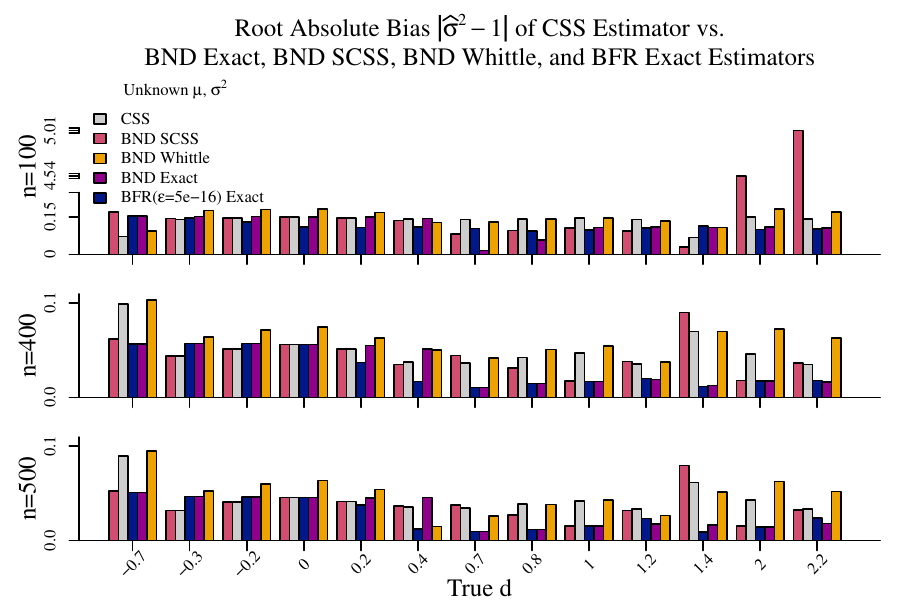}
\label{fig:ssecomp}
\caption{Average absolute bias of estimates of the noise variance $\sigma^2$ across $100$ simulated ARFIMA$\left(0, d, 0\right)$ time series with mean $\mu = 0$ and variance $\sigma^2 = 1$ for each sample size $n$ and true value of the differencing parameter $d$.}
\end{figure}
\clearpage

\section{Estimation of $d$ Comparison with Alternatives When Estimating a Polynomial Mean}\label{appsec:sec5}

\begin{figure}[h]
\centering
\includegraphics{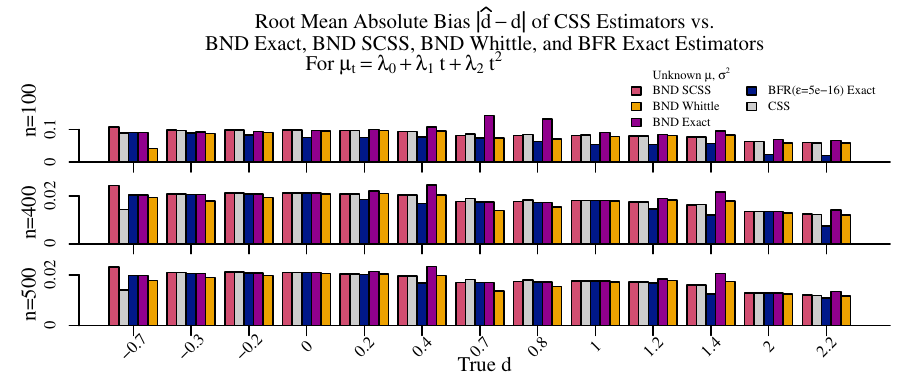}
\label{fig:biasvsfixed}
\caption{Root average absolute bias of CSS, BND, and BFR$\left(5\times 10^{-16}\right)$ estimators of $d$ across 100 simulated mean-zero ARFIMA$\left(0, d, 0\right)$ time series with variance $\sigma^2 = 1$. Estimates are obtained assuming $\mu_t =  \sum_{j = 0}^2 t^j \lambda_j$, where $\lambda_0$, $\lambda_1$, $\lambda_2$ and $\sigma^2$ are treated as unknown.}
\end{figure}
\clearpage

\section{Estimation of $d$ Comparison with Alternatives for Heavy Tailed Observations}\label{appsec:sec6}

\begin{figure}[h]
\centering
\includegraphics{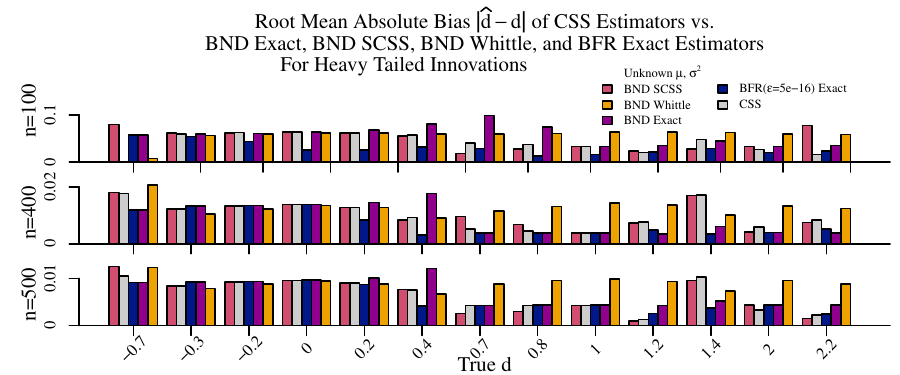}
\label{fig:biasvslight}
\caption{Root average absolute bias of CSS, BND, and BFR$\left(5\times 10^{-16}\right)$ estimators of $d$ across 100 simulated mean-zero time series simulated according to $\left(1 - B\right)^d y_t = z_t$, where $z_t$ are independent, identically distributed mean zero and unit variance Laplace random variables, which have heavier-than-normal tails. Estimates are obtained assuming $\mu_t =  \mu$ where $\mu$ and $\sigma^2$ are treated as unknown.}
\end{figure}
\clearpage

\section{Estimation of $d$ Comparison with Alternatives for Light Tailed Observations}\label{appsec:sec7}

\begin{figure}[h]
\centering
\includegraphics{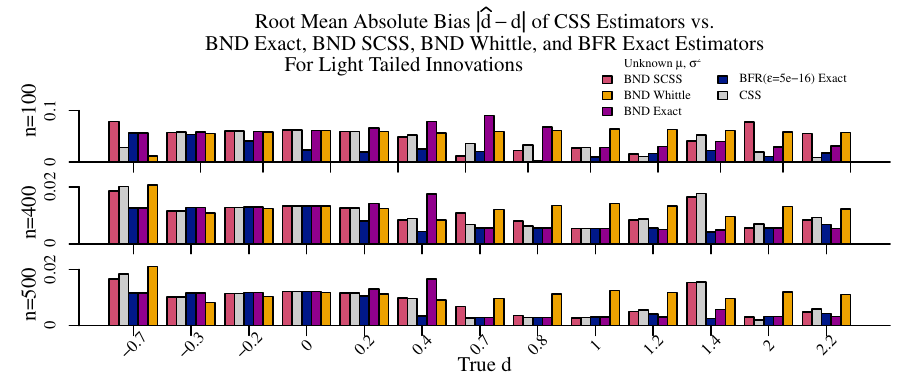}
\label{fig:biasvslight}
\caption{Root average absolute bias of CSS, BND, and BFR$\left(5\times 10^{-16}\right)$ estimators of $d$ across 100 simulated mean-zero time series simulated according to $\left(1 - B\right)^d y_t = z_t$, where $z_t$ are independent, identically distributed mean zero and unit variance generalized normal random variables with shape parameter 6, which have lighter-than-normal tails. Estimates are obtained assuming $\mu_t =  \mu$ where $\mu$ and $\sigma^2$ are treated as unknown.}
\end{figure}
\clearpage

\section{Estimation of Confidence Intervals}\label{appsec:sec8}

\begin{figure}[h!]
\centering
\includegraphics{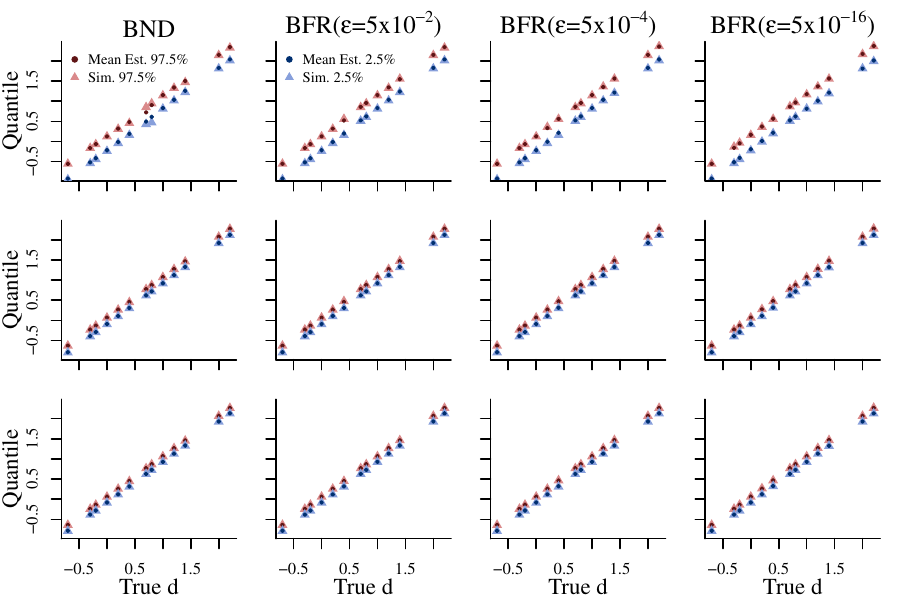}
\label{fig:checkint}
\caption{Comparison of lower 2.5\% and upper 97.5\% quantiles of BND and BFR$\left(\epsilon\right)$ exact likelihood estimators obtained by simulating $100$ simulated ARFIMA$\left(0, d, 0\right)$ time series with mean $\mu = 0$ and variance $\sigma^2 = 1$ versus average lower 2.5\% and average upper 97.5\% quantiles corresponding to BND and BFR$\left(\epsilon\right)$ exact likelihood estimators obtained using numerical differentiation across $100$ simulated ARFIMA$\left(0, d, 0\right)$ time series with mean $\mu = 0$ and variance $\sigma^2 = 1$ for each sample size $n$ and true value of the differencing parameter $d$.}
\end{figure}

\begin{figure}[h] 
\centering
\includegraphics{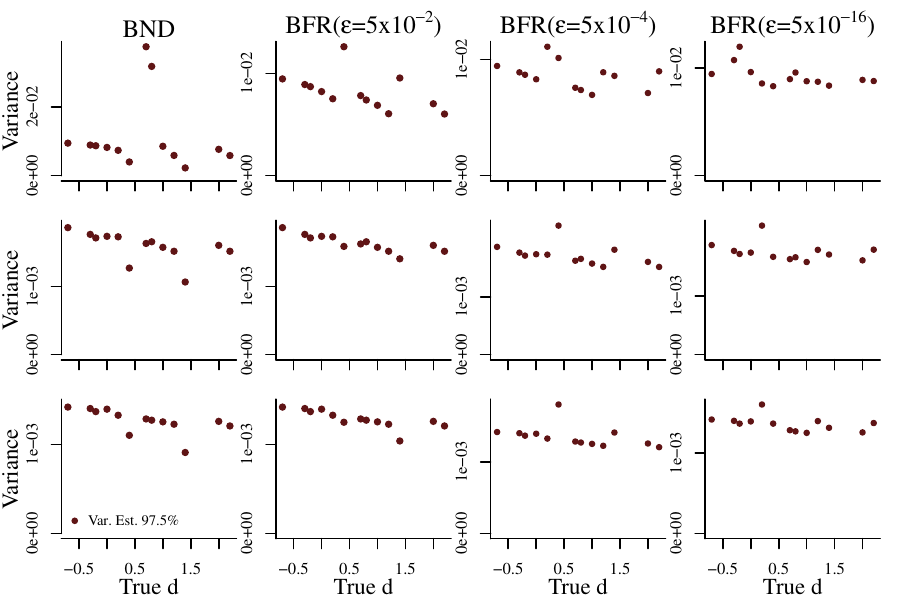}
\label{fig:checkint}
\caption{Variance of upper 97.5\% quantiles corresponding to BND and BFR$\left(\epsilon\right)$ exact likelihood estimators obtained using numerical differentiation across $100$ simulated ARFIMA$\left(0, d, 0\right)$ time series with mean $\mu = 0$ and variance $\sigma^2 = 1$ for each sample size $n$ and true value of the differencing parameter $d$.}
\end{figure}
\clearpage
\section{Likelihood Instability}\label{appsec:sec9}

\begin{figure}[h] 
\centering
\includegraphics{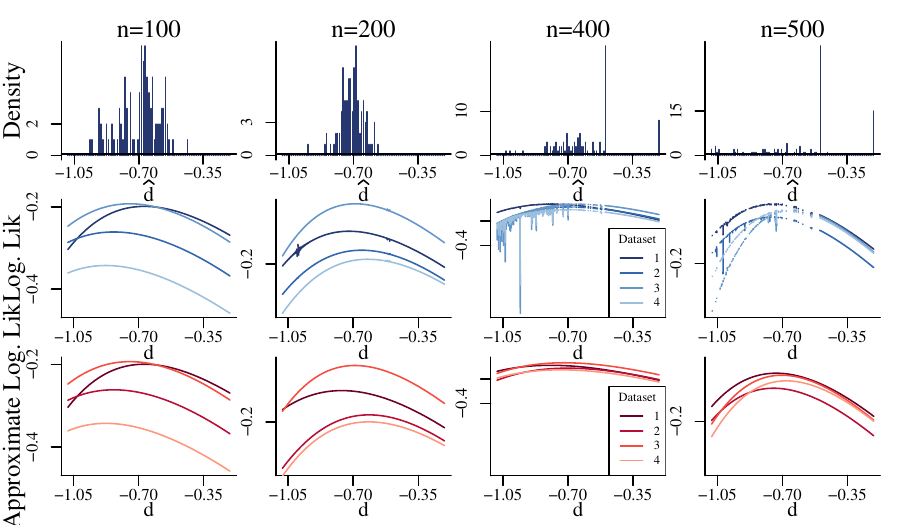}
\label{fig:trouble}
\caption{The first row shows histograms of exact maximum likelihood estimates $\hat{d}$ of the differencing parameter $d$ obtained by setting $\bar{d} = 3.5$ across a subset of $100$ simulated ARFIMA$\left(0, -0.7, 0\right)$ time series with mean $\mu = 0$ and variance $\sigma^2 = 1$ for each sample size $n$. The second row shows selected profile log-likelihood curves obtained by setting $\bar{d} = 3.5$ for four of the time series depicted in the previous row for each value of $n$. The third row shows approximate profile log-likelihood curves corresponding to the profile log-likelihood curves shown in the previous row obtained by treating values of the time series as conditionally independent after conditioning on the previous 100 values.}
\end{figure}

\clearpage
\section{Chemical Process Concentration and Temperature}\label{appsec:sec10}

\begin{figure}[h]
\centering
\includegraphics{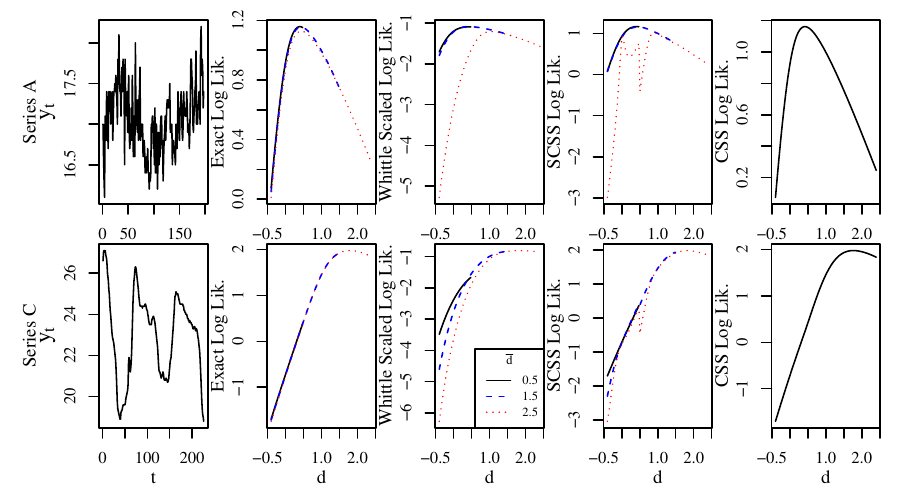}
\caption{Observed Series A and Series C time series and corresponding exact and Whittle profile  log-likelihood curves for $\bar{d} \in \left\{0.5, 1.5, 2.5\right\}$, with the mean $\mu$ and variance $\sigma^2$ profiled out.}
\label{fig:seriesacnoar}
\end{figure}

\begin{table}
\footnotesize
\centering
\begin{tabular}{ccc|rc|c|c|c}
 & & & \multicolumn{2}{c|}{Exact} &  \multicolumn{1}{c|}{Whittle} & \multicolumn{1}{c|}{SCSS}  & \multicolumn{1}{c}{CSS} \\
Data & $n$ & $\bar{d}$ & \multicolumn{1}{c}{$\hat{d}_{\bar{d}}$}  & \multicolumn{1}{c|}{95\% Interval for $d$}  & \multicolumn{1}{c|}{$\hat{d}_{\bar{d}}$} & \multicolumn{1}{c|}{$\hat{d}_{\bar{d}}$}  & \multicolumn{1}{c}{$\hat{d}_{\bar{d}}$}\\ \hline 
\multirow{3}{*}{Series A} & \multirow{3}{*}{$197$} & $0.5$ & $0.400$ & $\left(0.304, 0.496\right)$ & \cellcolor{lightgray}  $0.420$ &  \cellcolor{lightgray} $0.418$ &  \multirow{3}{*}{$0.418$} \\ 
& & $1.5$ &\cellcolor{lightgray}$0.427$ & \cellcolor{lightgray}$\left(0.319, 0.534\right)$ & $0.422$  & $0.500$ & \\
& & $2.5$ & $0.436$ &  $\left(0.326, 0.545\right)$ & $1.047$ & $0.901$  &  \\ \hline
\multirow{3}{*}{Series C} & \multirow{3}{*}{$226$} &  $0.5$ & $0.500$  & $-$& $0.500$ & $0.500$ & \multirow{3}{*}{$1.766$}\\ 
& & $1.5$ & $1.500$ &  $-$ & $1.500$ & $1.500$ & \\
& &$2.5$  & \cellcolor{lightgray} $1.788$ & \cellcolor{lightgray}$\left(1.659,    1.918\right)$ & \cellcolor{lightgray} $1.799$ &  \cellcolor{lightgray} $1.810$ 
\\ \hline
\end{tabular}
\caption{Estimates and corresponding 95\% confidence intervals for $d$ for the chemical process concentration readings (Series A) and chemical process temperature readings (Series C) for different values of $\bar{d}$. The BFR$\left(\epsilon =5\times 10^{-16}\right)$ exact likelihood,  BND Whittle, and BND SCSS estimates are highlighted in gray. 95\% intervals for exact likelihood estimates $\hat{d}_{\bar{d}}$ are provided for values of $\bar{d}$ that correspond to log-likelihoods that are decreasing at $\bar{d}$.  The mean $\mu$ and variance $\sigma^2$ are treated as unknown when estimating the differencing parameter $d$. }
\label{tab:apppe}
\end{table}

\begin{sidewaystable}
\centering
\tiny
\begin{tabular}{ccc|ccc|ccc|ccc|ccc}
 & & & \multicolumn{3}{c|}{Exact} & \multicolumn{3}{c|}{Whittle} & \multicolumn{3}{c|}{SCSS} & \multicolumn{3}{c}{CSS} \\
Data & $n$ & $\bar{d}$ & $d$ & $\theta_1$ & $\phi_1$ & $d$ & $\theta_1$ & $\phi_1$ & $d$ & $\theta_1$ &  $\phi_1$  & $d$ & $\theta_1$ &  $\phi_1$ \\ \hline 
\multirow{4}{*}{Series A} & \multirow{4}{*}{$197$} & $0.5$ & $0.419$ & $-0.037$ & $-$ &\cellcolor{lightgray} $0.449$ & \cellcolor{lightgray}$-0.043$ & \cellcolor{lightgray}$-$ &\cellcolor{lightgray}0.480& \cellcolor{lightgray}$-0.093$ & \cellcolor{lightgray} $-$ & \multirow{4}{*}{$0.479$} & \multirow{4}{*}{$-0.092$} & \multirow{4}{*}{$-$} \\ 
& & $1.5$ 
&\cellcolor{lightgray}$0.502$  & \cellcolor{lightgray}$-0.117$ &  \cellcolor{lightgray}$-$ &$0.470$ & $-0.077$ & $-$ &  $0.500$ &  $-0.115$ &  $-$ \\
& & $2.5$ & $1.314$  & $-0.923$ & $-$ & $2.047$& $-1.000$  & $-$ &  $0.990$ &  $-0.685$ &  $-$  \\ 
& & $3.5$ & $1.310$ & $-0.911$ & $-$ & $2.292$& $-0.665$ & $-$ &  $0.994$ &  $-0.721$ &  $-$  \\ \hline
\multirow{4}{*}{Series C} & \multirow{4}{*}{$226$} &  $0.5$ 
& $0.500$ & $-$& $1.000$& $0.215$  & $-$  & $0.907$ &  $0.500$ &  $-$ & $0.982$ &  \multirow{4}{*}{$0.939$} & \multirow{4}{*}{$-$} & \multirow{4}{*}{$0.857$} \\ 
& & $1.5$ & $0.950$ & $-$ & $0.850$ &\cellcolor{lightgray}$0.908$ &\cellcolor{lightgray}$-$  &\cellcolor{lightgray}$0.858$ &\cellcolor{lightgray} $0.894$ &\cellcolor{lightgray} $-$ & $\cellcolor{lightgray}0.881$ \\
& &$2.5$  & \cellcolor{lightgray}$0.972$ & \cellcolor{lightgray}$-$ &  \cellcolor{lightgray}$0.842$ &$0.892$ &$-$ &$0.930$ &  $1.014$ &  $-$ & $0.815$   \\
& & $3.5$  & $0.971$ &$-$ & $0.852$ & $1.109$ &$-$  &$1.000$  &  $0.996$ &  $-$ & $0.828$ 
\\ \hline
\end{tabular}
\caption{Estimates of the parameters  of  ARFIMA$\left(0, d, 1\right)$ and ARFIMA$\left(1, d, 0\right)$  models both with $\mu_t = \mu$     for the chemical process concentration readings (Series A) and chemical process temperature readings (Series C) for different values of $\bar{d}$. The BFR$\left(\epsilon =5\times 10^{-16}\right)$ exact likelihood estimates, BND Whittle estimates, and BND SCSS estimates  are highlighted in gray. The mean $\mu$ and variance $\sigma^2$ are treated as unknown when estimating the differencing parameter $d$. }
\label{tab:armaapppe}
\end{sidewaystable}

\begin{figure}
\centering
\includegraphics{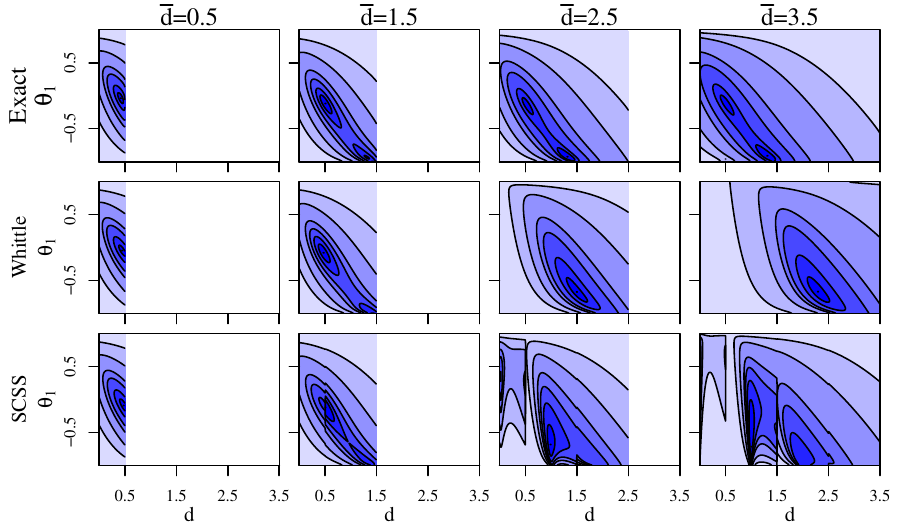}
\caption{Exact, Whittle, and SCSS  joint  profile  log-likelihoods for Series A as a function of the  moving average parameter $\theta_1$  and $d$ for $\bar{d} \in \left\{0.5, 1.5, 2.5, 3.5\right\}$, with the mean $\mu$ and variance $\sigma^2$ profiled out.}
\label{fig:seriesaar}
\end{figure}

\begin{figure}
\centering
\includegraphics{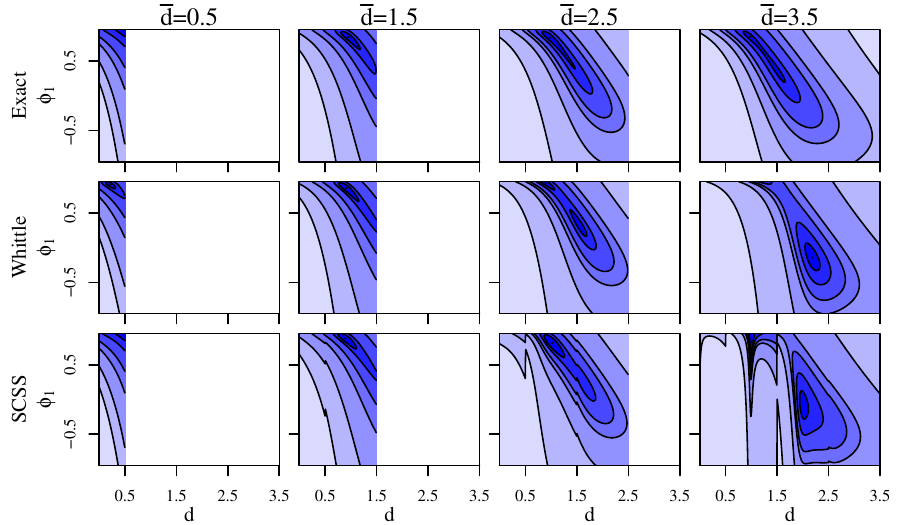}
\caption{Exact, Whittle, and SCSS  joint  profile  log-likelihoods for Series C as a function of the autoregressive parameter $\phi_1$ and $d$ for $\bar{d} \in \left\{0.5, 1.5, 2.5, 3.5\right\}$, with the mean $\mu$ and variance $\sigma^2$ profiled out.}
\label{fig:seriescar}
\end{figure}

\begin{figure}
\centering
\includegraphics{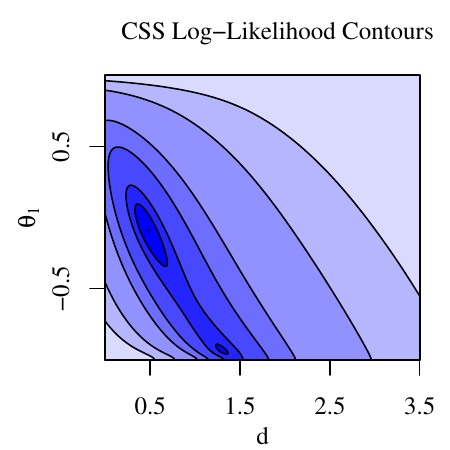}
\caption{CSS joint profile log-likelihood for Series A as a function of the moving average parameter $\theta_1$ and $d$, with the mean $\mu$ and variance $\sigma^2$ profiled out.}
\label{fig:seriesacssnodiff}
\end{figure}

\begin{figure}
\centering
\includegraphics{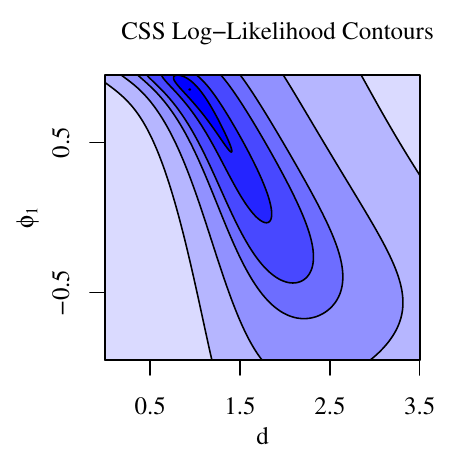}
\caption{CSS joint profile log-likelihood for Series C as a function of the autoregressive parameter $\phi_1$ and $d$, with the mean $\mu$ and variance $\sigma^2$ profiled out.}
\label{fig:seriesccssnodiff}
\end{figure}

\clearpage
\section{Truncated vs. Untruncated Fractional Differences}\label{appsec:sec11}

The model that assumes that truncated fractional differences $\left(1 - B\right)^d_+ y_t$ are  distributed according to a stationary ARMA$\left(p, q\right)$ model and the model that assumes that untruncated fractional differences $\left(1 - B\right)^d y_t$ are distributed according to a stationary ARMA$\left(p, q\right)$ model can both be represented as $y_{t + 1} = \sum_{j = 1}^{t} c_{tj} y_{t+1-j} + z_{t + 1}$, where $z_t$ are independent, mean zero random variables with variance $v_t$ and the values of $c_{tj}$ and $v_t$ are one-step-ahead forecast coefficients and variances determined by which model is being used and its parameters. Comparing the one-step-ahead forecast coefficients and variances across the two models allows us to compare the dependence structure of data generated under the two models. 
We consider the simpler case where $p = q = 0$ and compute the ratios $c_{tj}/\sqrt{v_t}$ for $d \in \{-0.5, -0.75, -1, -1.25, -1.5, -1.75\}$ and $t \in \{0, 1, 2, 3\}$. The ratios $c_{tj}/\sqrt{v_t}$ are shown in Figure~\ref{fig:ctss}. They are similar when $d = -0.5$ but diverge substantially as $d$ decreases. This suggests that when $d \leq -1$, the two models are very different.

\begin{figure}[h] 
\centering
\includegraphics{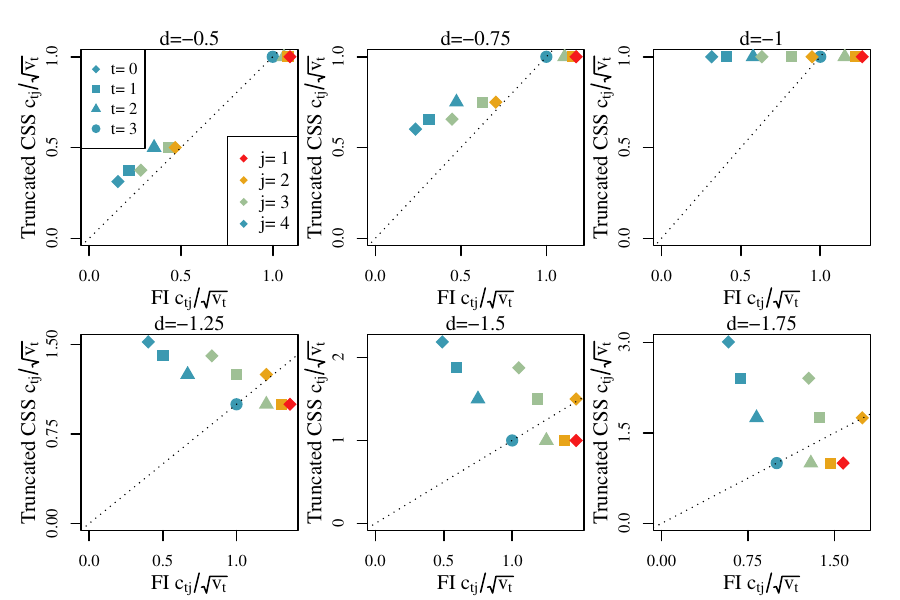}
\label{fig:ctss}
\caption{Each panel compares the scaled one-step-ahead forecast coefficients $c_{tj}/\sqrt{v_t}$ for the truncated and untruncated fractional difference models for values of $d$ between $-0.5$ and $-1.75$.}
\end{figure}

\clearpage
\section{Series C Data Residuals}\label{appsec:sec12}

\begin{figure}[h] 
\centering
\includegraphics{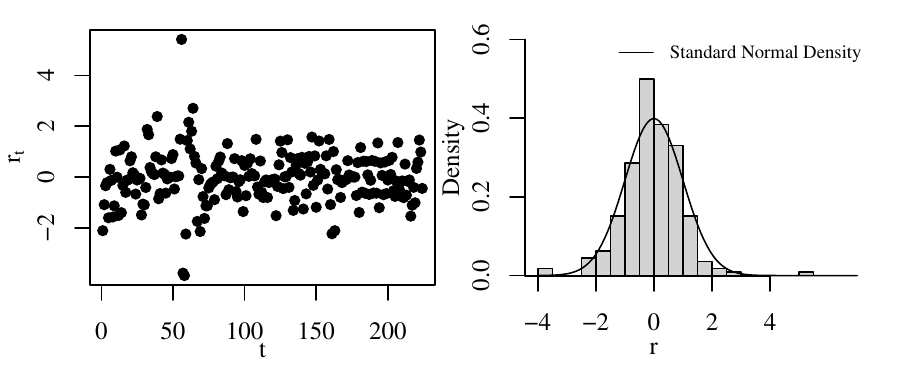}
\label{fig:residsc}
\caption{Residuals from fitting an ARFIMA$(0, d, 0)$ model with $\bar{d} = 2.5$ for Series C, using the estimated differencing parameter shown in Table 1, obtained by premultiplying the differenced response $\boldsymbol x^{(m_{\bar{d}})}$ with $m_{\bar{d}} = 2$ by the square root of the inverse of the estimated covariance matrix of the differenced deviations. The left panel plots the residuals in the order that they appear and the right panel plots a histogram of the residuals and compares them to a standard normal density.}
\end{figure}

\clearpage
\section{Series C Data Compared to Simulations}\label{appsec:sec13}

\begin{figure}[h] 
\centering
\includegraphics{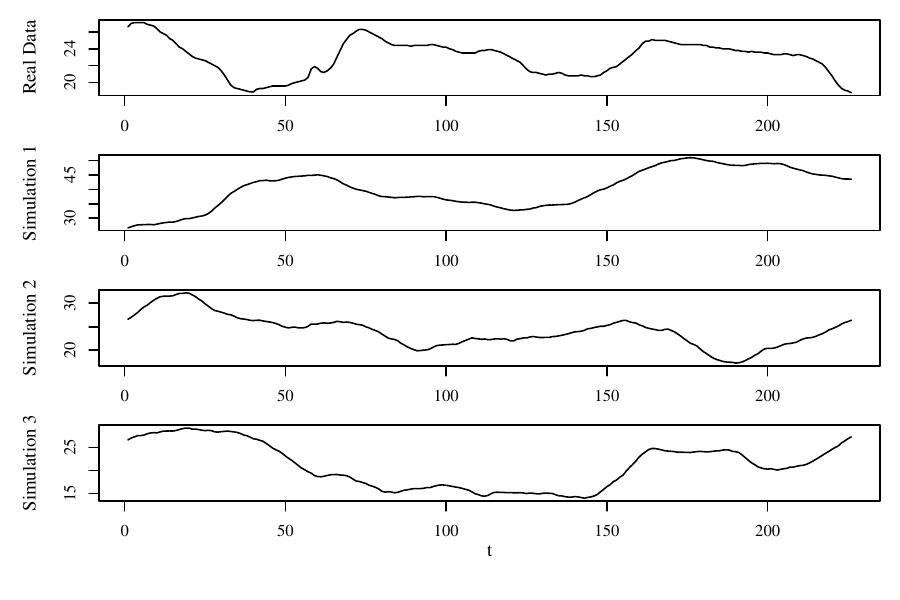}
\label{fig:simsc}
\caption{Comparison of the observed Series C data, plotted in the top panel, with simulated data of the same length simulated from the estimated ARFIMA$(0, d, 0)$ model with $\bar{d} = 2.5$ fit to the Series C data. Three random draws are depicted, to show the variability simulated time series from the fitted model.}
\end{figure}

\clearpage
\section{CO$_2$ Emissions Data}\label{appsec:sec14}

\begin{figure}[h] 
\vspace{-10mm}
\centering
\includegraphics{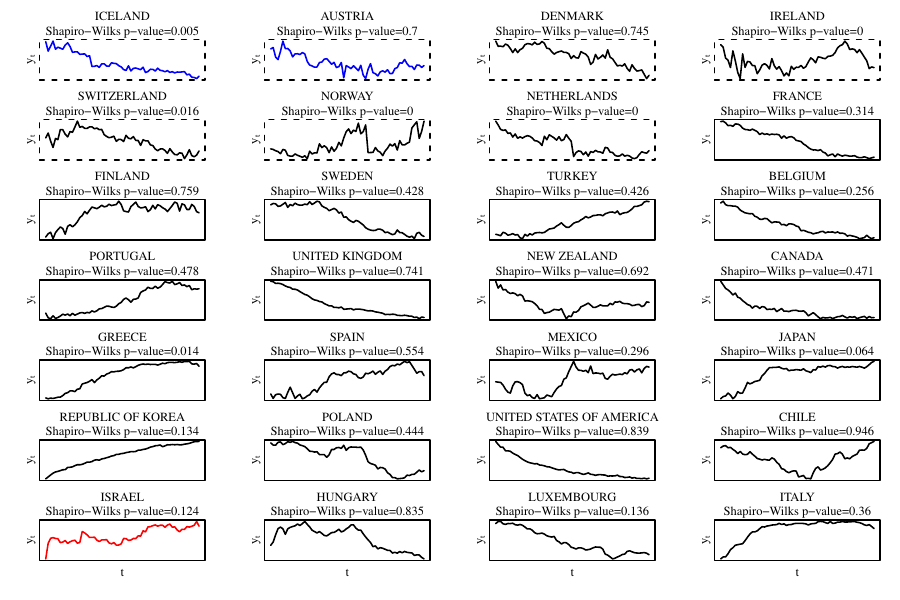}
\label{fig:co2obs}
\caption{Observed emissions time series and $p$-values for Shapiro-Wilks tests of normality of the residuals. Dashed boxes denote countries for which our proposed method rejects the null hypothesis that emissions are not mean reverting. Blue, black, and red lines correspond to the countries for which our proposed method chooses $\bar{d} = 1.5$, $\bar{d} = 2.5$, and $\bar{d} = 3.5$, respectively.}
\end{figure}

\clearpage
\section{Observed and Differenced CO$_2$ Data}\label{appsec:sec15}

\begin{figure}[h] 
\vspace{-10mm}
\centering
\includegraphics{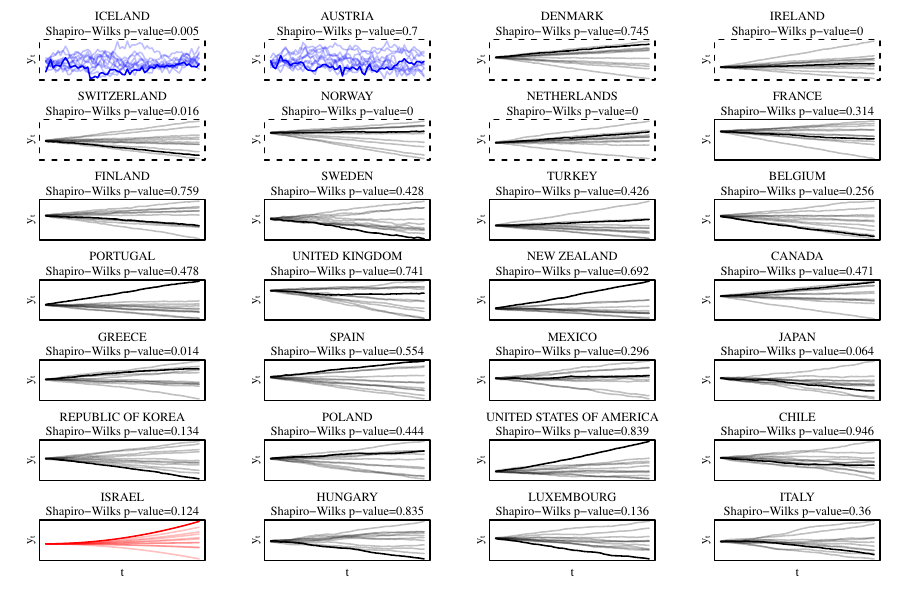}
\label{fig:simco2}
\caption{Observed detrended emissions time series (dark lines), simulated detrended time series from the fitted model (lighter lines), and $p$-values for Shapiro-Wilks tests of normality of the residuals. Dashed boxes denote countries for which our proposed method rejects the null hypothesis that emissions are not mean reverting. Blue, black, and red lines correspond to the countries for which our proposed method chooses $\bar{d} = 1.5$, $\bar{d} = 2.5$, and $\bar{d} = 3.5$, respectively.}
\end{figure}

\clearpage

\clearpage
\section{Selected ECIS Time Series}\label{appsec:sec16}

\begin{figure}[h] 
\centering
\includegraphics{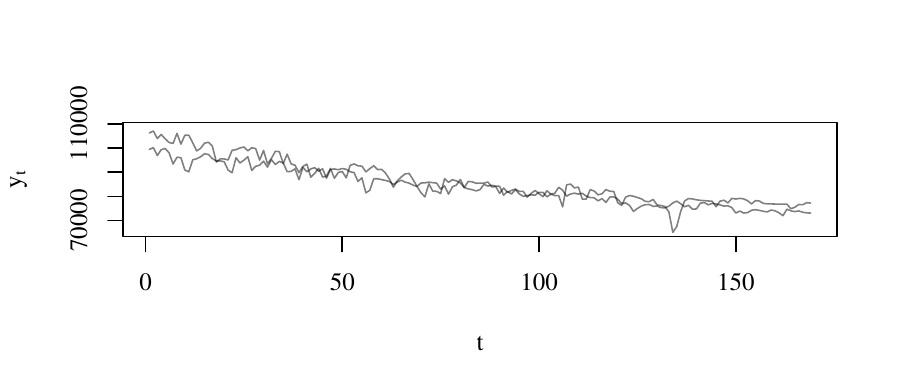}
\label{fig:mdckct}
\caption{Two selected time series that were used to obtain estimates presented in Figure 10 of the main text. These are both based on uncontaminated MDCK cells prepared using BSA measured at the lowest frequency during the first experiment. The $p$-value of a Shapiro-Wilks test of the null hypothesis of normality of the residuals for all replicates in this condition is  less than $10^{-15}$.}
\end{figure}

\begin{figure}[h] 
\centering
\includegraphics{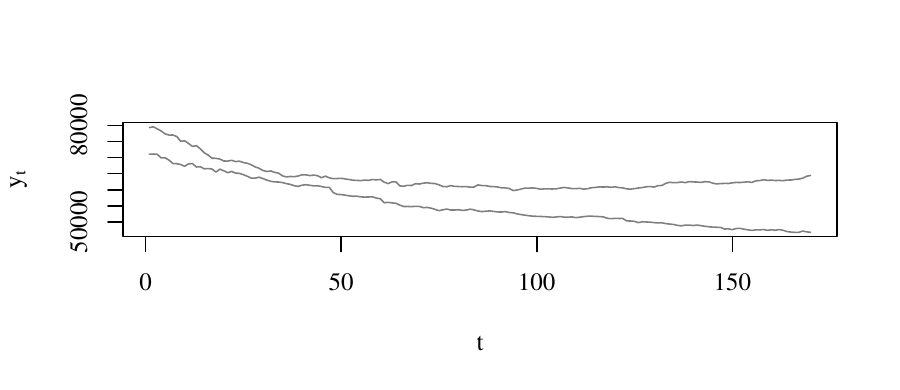}
\label{fig:mdckin}
\caption{Two selected time series that were used to obtain estimates presented in Figure 10 of the main text. These are both based on contaminated MDCK cells prepared using BSA measured at the lowest frequency during the first experiment. The $p$-value of a Shapiro-Wilks test of the null hypothesis of normality of the residuals for all replicates in this condition is  less than $10^{-15}$.}
\end{figure}

\clearpage

\begin{figure}[h] 
\centering
\includegraphics{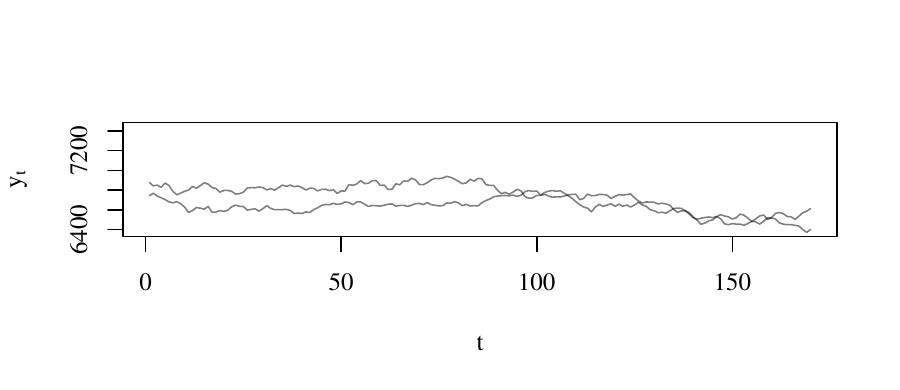}
\label{fig:bscct}
\caption{Two selected time series that were used to obtain estimates presented in Figure 10 of the main text. These are both based on uncontaminated BSC cells prepared using BSA measured at the highest frequency during the first experiment. The $p$-value of a Shapiro-Wilks test of the null hypothesis of normality of the residuals for all replicates in this condition is  less than $0.005$.}
\end{figure}

\begin{figure}[h] 
\centering
\includegraphics{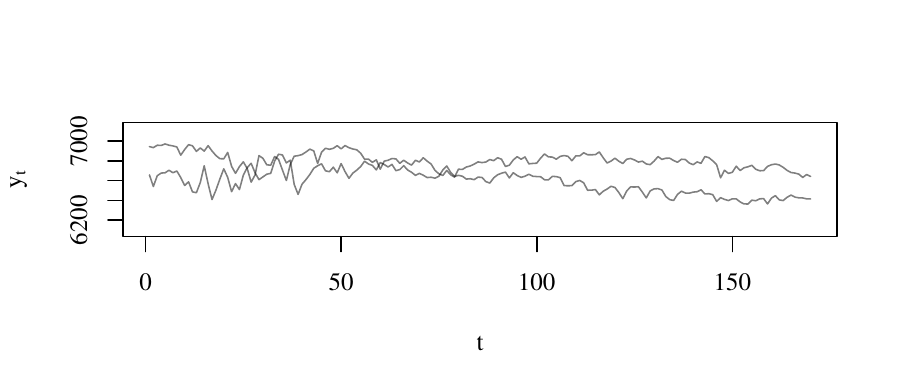}
\label{fig:bscin}
\caption{Two selected time series that were used to obtain estimates presented in Figure 10 of the main text. These are both based on contaminated BSC cells prepared using BSA measured at the lowest frequency during the first experiment. The $p$-value of a Shapiro-Wilks test of the null hypothesis of normality of the residuals for all replicates in this condition is  less than $10^{-15}$.}
\end{figure}

\end{appendix}
	
\end{document}